\documentclass[11pt]{article}
\usepackage{graphicx}
\usepackage{pgfplots}
\usepackage{subcaption}
\usepackage[font=scriptsize]{caption} 

\usepackage[colorlinks=true,citecolor=blue]{hyperref}
\usepackage{natbib}
\usepackage{graphicx}
\usepackage{amsfonts}
\usepackage{amsmath}
\usepackage{amssymb}
\usepackage{url}
\usepackage{fancyhdr}
\usepackage{indentfirst}

\usepackage{enumerate}
\usepackage{titlesec}
\usepackage{amsthm}
\usepackage{dsfont}
\usepackage[misc]{ifsym}

\theoremstyle{definition}

\newtheorem{example}{Example}

\newtheorem{definition}{Definition}
\newtheorem{assumption}{Assumption}

\theoremstyle{plain}
\newtheorem{theorem}{Theorem}

\newtheorem{proposition}{Proposition}

\theoremstyle{remark}
\newtheorem{remark}{Remark}

\usepackage{cases}

\theoremstyle{definition}

\def\N{\mathbb{N}}

\def\p{\mathbb{P}}
\def\E{\mathbb{E}}

\def\R{\mathbb{R}}

\def\M{\mathcal{M}}

\def\X{\mathcal{X}}

\def\var{\mathrm{Var}}

\def\X{\mathcal{X}}

\def\d{\,\mathrm{d}}
\DeclareMathOperator*{\esssup}{ess\text{-}sup}

\newcommand{\GD}{\mathrm{GD}}
\newcommand{\GC}{\mathrm{GC}}

\DeclareMathOperator*{\argmin}{arg\,min}

\usepackage[onehalfspacing]{setspace}

\usepackage{bm}
\usepackage{tikz-qtree}
\usepackage{tikz}

\pgfplotsset{compat=1.18}

\title{Higher-order Gini indices: An axiomatic approach}
\author{Xia Han\thanks{School of Mathematical Sciences, LPMC and AAIS, Nankai University, China.  \Letter~{\url{xiahan@nankai.edu.cn}}}  \and Ruodu Wang\thanks{Department of Statistics and Actuarial Science, University of Waterloo,  Canada. \Letter~{\url{wang@uwaterloo.ca}}}\and  Qinyu Wu\thanks{Department of Statistics and Actuarial Science, University of Waterloo, Canada.  \Letter~{\url{q35wu@uwaterloo.ca}}}}
\date{\today}

\begin{document}
	\maketitle
	\begin{abstract}

Via an axiomatic approach, we  characterize the family of $n$-th order Gini deviation, defined as the expected range over $n$ independent draws from a distribution, to quantify joint dispersion across multiple observations. This family extends the classical Gini deviation, which relies solely on pairwise comparisons.
The normalized version is called a high-order Gini coefficient. 
The generalized indices grow increasingly sensitive to tail inequality as $n$ increases, offering a more nuanced view of distributional extremes.  The higher-order Gini deviations admit a Choquet integral representation, inheriting the desirable properties of coherent deviation measures. Furthermore, we show that both the $n$-th order Gini deviation and the $n$-th order Gini coefficient are statistically $n$-observation elicitable, allowing for direct computation through empirical risk minimization. Data analysis using World Inequality Database data reveals that higher-order Gini coefficients  capture disparities that the classical Gini coefficient may fail to reflect, particularly in cases of extreme income or wealth concentration. 
\end{abstract}


 \noindent\rule{\textwidth}{0.5pt}

\section{Introduction}

The Gini deviation (GD) and Gini coefficient (abbreviated as GC, also called the Gini index) are fundamental tools for measuring inequality and dispersion in economic research, owing to its intuitive interpretation, broad comparability, and robust statistical properties.
Originally introduced by \cite{G12, G21}, the classical GC is defined through its geometric definition based on the Lorenz curve.  Its intuitive interpretation and straightforward quantile representation have established it as a standard summary statistic (\cite{G72} and \cite{LY84}).  Beyond its theoretical appeal, GC is widely used in practice to summarize and compare income or wealth distributions across countries and over time (\cite{KSS10}). Policymakers rely on it to monitor inequality trends, evaluate the effects of tax and transfer systems, and design social welfare programs (\cite{D00} and \cite{KP05}). For example, an increase in GC is often interpreted as a signal of rising income concentration, motivating debates over progressive taxation or redistribution. In international development, organizations such as the World Bank and the Organisation for Economic Co-operation and Development  use the Gini index as a headline indicator to track progress toward reducing inequality and promoting inclusive growth.  
The measurement of economic inequality has been extensively treated in numerous monographs and handbooks, with GC and its various extensions occupying a central position in this literature; see, e.g., \cite{AP07}, \cite{BD11} and \cite{AB14}. Beyond public policy, Gini-type indices, especially GD, are employed in actuarial risk measurement (\cite{D90}, \cite{FZ17} and \cite{FWZ17}),  decision analysis (\cite{EL21}, who called GD as maxiance), portfolio optimization (\cite{Y82} and \cite{RGC04}), and statistical scoring rules (\cite{ZDT00} and \cite{G11}).

Both GD and GC can be formulated as functions of random variables. For a random variable $X$ that represents a distribution of wealth or income,
GD averages relative dispersion through pairwise differences:
\begin{equation}
    \label{eq:intro-GD}
\GD(X) = \frac{1}{2}\, \E[|X - X'|] , 
\end{equation} 
where $X'$ is an independent and identically distributed (iid) copy of $X$. 
The GC of $X$ is defined as
$$
\GC(X) = \frac{\GD(X)}{\E[X]} = \frac{\E[|X - X'|]}{\E[X+X']},
$$
where we assume $X\ge 0$ and $\E[X]>0$. Both GD and GC use two iid observations to capture the inequality via all possible pairwise contrasts in the distribution.

\begin{figure}[t!]
\centering
 \includegraphics[width=12cm]{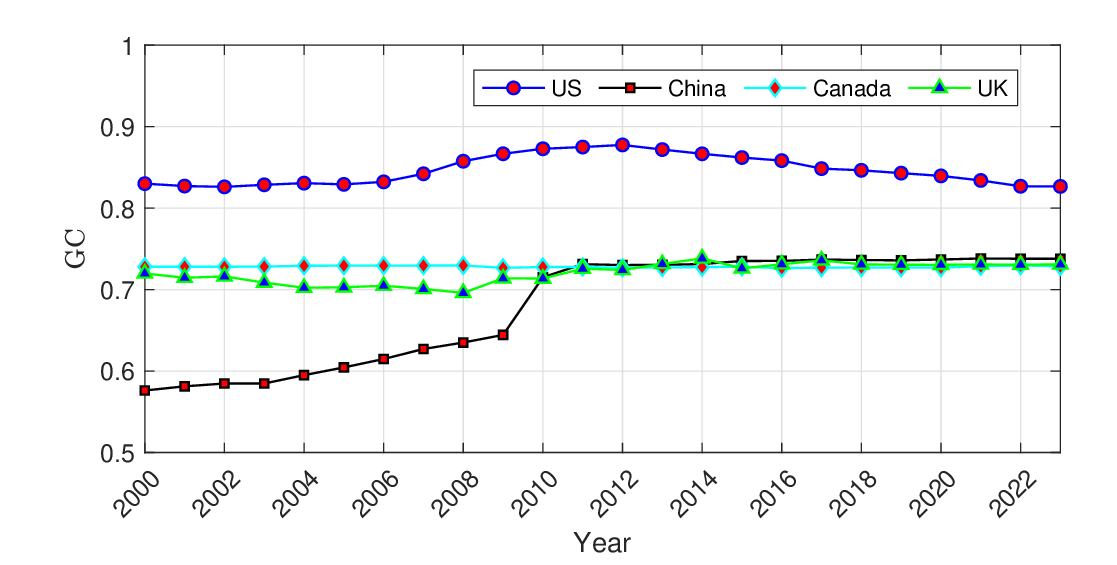}
 \captionsetup{font=small}
 \caption{\small Values of $\mathrm{GC}$ based on household wealth for the United States, China, Canada, and the United Kingdom from 2000 to 2023
}\label{fig:GC2_country}
\end{figure}
Despite its popularity, GC is only one measure of distributional inequality or dispersion, and it may fail to report some practically important features of the underlying distribution.\footnote{The measurement of inequality is a complex task. For instance, \cite{CK25} provided a comprehensive discussion on various concepts and definitions in measuring income inequality. Our paper focuses on the choice of the indices as measurement tool, implicitly assuming that the quantity to be measured (e.g., income or wealth)  is well defined.} 
Figure \ref{fig:GC2_country} presents the GC curves of wealth distributions for four economies---US, China (mainland), UK, and Canada---over the period of 2000 to 2023, with data from the World Inequality Database (WID).\footnote{WID (\url{http://wid.world}) provides the full wealth distribution for each country, but not the full income distribution for each country. Therefore, we take the wealth distributions as the main example.} As we can see from the figure, 
US has the largest GC among the four countries, and UK, Canada and China have similar GC curves
since 2010, although the GC curve of China is substantially smaller than the others before 2010.  
From these curves, 
one may be tempted to make the following conclusions.
\begin{enumerate}
    \item  
China has a similar wealth distribution to Canada and UK since 2010. 
    \item The wealth distribution in China has a similar level of equality to Canada and UK since 2010. 
    \item The wealth inequality of China has been increasing up to 2010, and it is stabilized after 2010.  
\end{enumerate}
To better understand these statements, we plot in 
Figure~\ref{fig:topshare-intro}  the wealth shares held by the top 1\% and top 10\% of the population in China, Canada, and UK. The data clearly show that both the top 1\% and top 10\% populations in China have consistently held a larger share of total wealth than their counterparts in Canada and the UK after 2010. This provides evidence against the first two  statements, indicating that China’s wealth distribution has been notably more unequal than UK and Canada. 
On the other hand, both Figures \ref{fig:GC2_country} and \ref{fig:topshare-intro} agree that the wealth distribution in China may be stabilized after 2010. 

\begin{figure}[t!]
\centering
 \includegraphics[width=17cm]{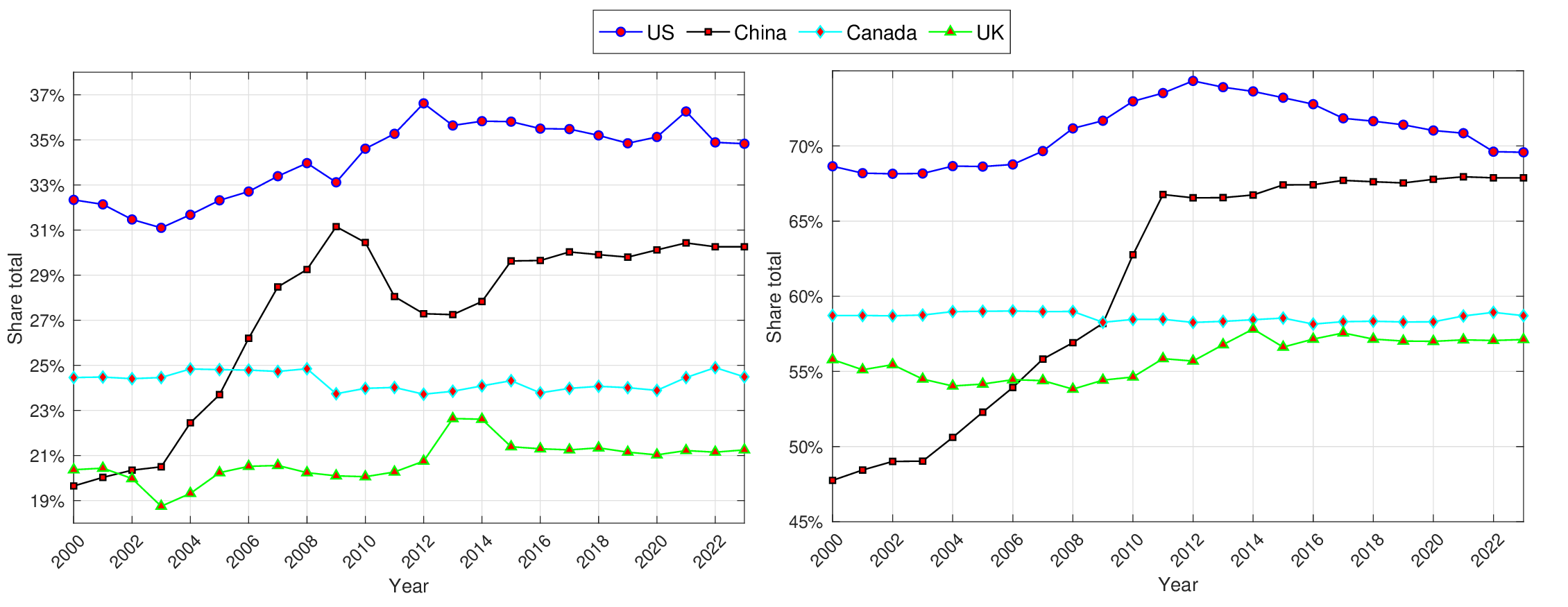}
 \captionsetup{font=small}
 \caption{\small Wealth shares held by the top 1\% (left) and top 10\% (right) in  the United States, China, Canada, and the United Kingdom from 2000 to 2023.
 }\label{fig:topshare-intro}
\end{figure}

The above discussion suggests that GC may not provide the most relevant information about wealth inequality, although it is the most popular index for this purpose. 
To address this limitation of GC, we consider generalizing 
GC to a family of indices. We begin by examining  key properties satisfied by GD and GC, which we hope to keep.   A key insight of GD in \eqref{eq:intro-GD} is that it can be expressed as the expected value of a function of two independent copies of the underlying variable.  Motivated by this observation, we consider natural extensions that involve expectations of functions of $n$ iid copies, with $n \geq 2$.
Building on this idea,  together with several other natural properties satisfied by GD---including symmetry, comonotonic additivity, and continuity---we develop an axiomatic characterization showing that any measure satisfying these axioms can be expressed as  a linear combination of functionals in the following family, which we call the \emph{higher-order Gini deviations},
$$
\GD_n(X) = \frac{1}{n}\, \E \bigl[\max\{X_1,\dots,X_n\} - \min\{X_1,\dots,X_n\}\bigr],
$$
where $X_1, \dots, X_n$ are iid copies of $X$. 
The proof of our axiomatic characterization goes by first establishing a signed Choquet integral representation and then identifying those that can be represented by expected values of functions of iid observations. 
Many classes of  decision models, such as the dual utility of \cite{Y87} and the Choquet expected utility model of \cite{Schmeidler89}, are formulated based on Choquet integrals;  see the more recent works of \cite{StrackWambach17}, \cite{GMSZ25} and \cite{HK25}.  The main differences here is that our functionals $\GD_n$ do not satisfy monotonicity, thus they are called ``signed".

Observing from \eqref{eq:intro-GD} that the Gini deviation admits the alternative representation $$\GD(X) = \frac{1}{2} \E[\max \{X,X'\} - \min \{X,X'\}],$$ it follows that $\GD_n$ with 
 $n=2$  coincides with GD. Thus, the $n$-th order Gini deviation $\GD_n$ replaces the pairwise absolute difference with the expected range of 
$n$ independent draws, thus quantifying the average spread within larger groups. As $n$ increases, $\GD_n$ becomes increasingly sensitive to tail behavior and concentration,  since the range over
$n$ iid draws amplifies the impact of extreme values. Based on $\GD_n$,
the corresponding \emph{$n$-th order  Gini coefficient} is defined by   $$\GC_n(X) = \frac{\GD_n(X) }{\mathbb{E}[X]}.$$ The quantity $\mathrm{GC}_n$ as a generalization of GC  was recently proposed by \cite{GROG24}, but a comprehensive study is still lacking, particularly with respect to an axiomatic framework and empirical analysis.

Figure~\ref{fig:GC10_country} depicts $\GC_{10}$ for the four countries. 
The values of $\mathrm{GC}_{10}$ in China increase rapidly during the considered period, surpassing those of Canada and UK, and gradually approaching the levels observed in US. This upward trend underscores the ability of $\mathrm{GC}_{10}$ to capture shifts in upper-tail wealth concentration, and it is consistent with the observations from Figure \ref{fig:topshare-intro} regarding statements 1--3 above.
Therefore, $\GC_n$ provides interpretable information that is complementary to the classic $\GC$, better capturing more nuanced features of the wealth distributions. 
A more comprehensive analysis in Section \ref{sec:8} illustrates $\GC_n$ with other values of $n$ and other datasets.

\begin{figure}[t!]
\centering
 \includegraphics[width=12cm]{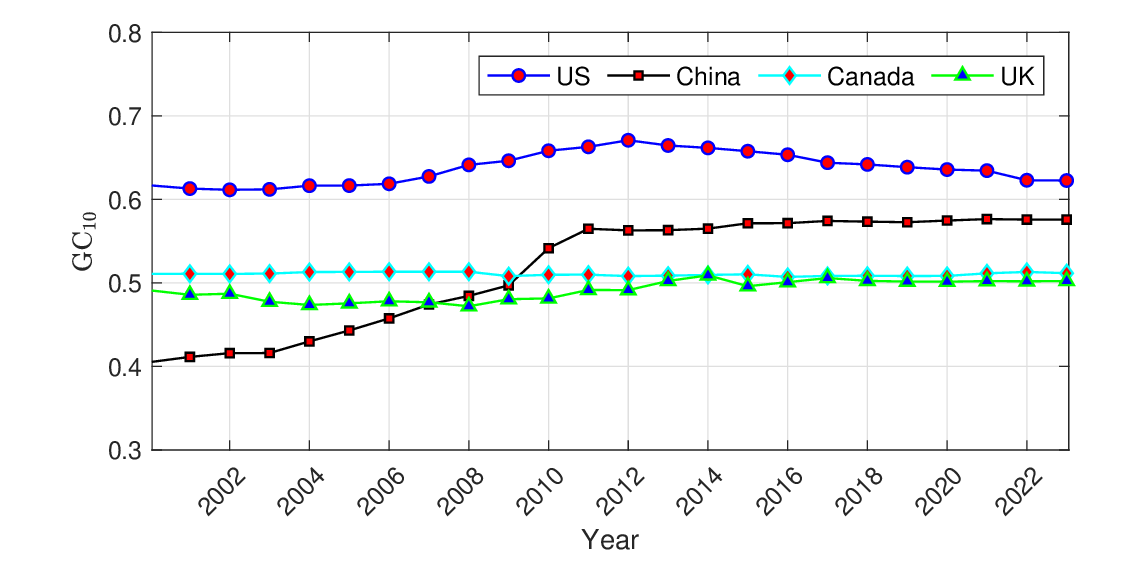}
 \captionsetup{font=small}
 \caption{\small Values of $\mathrm{GC}_{10}$ based on household wealth for the United States, China, Canada, and the United Kingdom from 2000 to 2023
}\label{fig:GC10_country}

\end{figure}

 
 Many other variants of GD and GC exist in the literature; see \cite{CV12} for several formulations of GC proposed by  Gini, and  \cite{Y98,Y03}, \cite{YS12} and \cite{VS25} for further developments.

After reviewing the classical Gini deviation and Gini coefficient in Section \ref{sec:2}, 
we first present the axioms satisfied by the Gini indices and then  develop a rigorous axiomatic framework for their higher-order generalizations  in Section \ref{sec:3}.  In Section \ref{sec:property}, we explore the key properties of the higher-order Gini indices in detail.   First, we examine the economic properties of the higher-order Gini coefficient, emphasizing its interpretability and practical relevance. Second,  we show that  $\GD_n$ admits a representation as a signed Choquet integral with a concave distortion function. 
Moreover, $\GD_n$ belongs to the class of deviation measures in the sense of \cite{RUZ06} and the class of coherent measures of variability in the sense of \cite{FWZ17}.
This integral formulation guarantees several desirable properties, including symmetry, convexity, and comonotonic additivity, and leads to economically useful properties of $\GC_n$, such as quasi-convexity and mixture-quasi-concavity.     
Finally, we  present a comparative analysis between the $n$-th order GD and the standard deviation. 

We study the multi-observation elicitability of $\GD_n$ and $\GC_n$ in Section \ref{sec:5}. 
Elicitability refers to the property of a statistical functional  being the unique minimizer of an expected scoring function, which is central to forecast evaluation and model comparison (\cite{G11} and \cite{FK21}). Many standard risk measures, including general deviation measures, are not elicitable in the classical (single-observation) framework (\cite{WW20}). To address this, the notion of multi-observation elicitability—where the score function depends on multiple independent observations—was introduced by \cite{CFMW17}; see also \cite{FMMW17}.  
Following from  their definitions as the expected value of functions of iid observations or their ratio,  
we show that both $\GD_n$ and $\GC_n$ are $n$-observation elicitable, thereby enabling rigorous comparative backtesting in risk and inequality measurement contexts. Section~\ref{sec:6} contains  results on the statistical inference of  $\GD_n$ and $\GC_n$ using its standard empirical estimators, including their consistency and asymptotic variance.

We do not intend to argue that the higher-order Gini indices $\mathrm{GC}_n$ necessarily outperform the classical case $\mathrm{GC}$ in all settings. Instead, we argue that higher-order indices can uncover features that $\mathrm{GC}$ fails to detect, thus providing complementary information, as explained in the example above.  
We further provide 
 more comprehensive empirical analyses using data from WID in Section \ref{sec:8}, spanning post-tax national income and household wealth across countries and regions between 2000 and 2023.  
 We find that higher-order Gini indices  such as $\mathrm{GC}_5$ or $\mathrm{GC}_{10}$  provide  additional insights in settings where  $\mathrm{GC}$ fails to distinguish some distributions.  These findings underscore the value of $\mathrm{GC}_n$ as a complementary tool for inequality assessment, especially in the context of evaluating tax policies, monitoring ultra-high-net-worth individuals, or studying long-term distributional trends.


Finally, Section~\ref{sec:9} concludes the paper. All proofs are put into Appendix \ref{sec:proofs}. Appendix~\ref{sec:7} contains explicit expressions for the higher-order Gini indices under several commonly used parametric distributions. The monotonicity properties of ${\rm GD}_n$ and ${\rm GC}_n$ with respect to the order $n$, along with simulation results of comparisons across different distributions and orders, are presented in Appendix \ref{sec:monotone}.

 \section{The Gini deviation and Gini coefficient}\label{sec:2}
We fix an atomless probability space $(\Omega,\mathcal F,\p)$.  
For $q\in [1,\infty)$, let
$L^q$ be the set of random variables with a   finite $q$th moment, 
and $L^\infty$ is the set of essentially bounded random variables.
Almost surely equal random variables are treated as identical.  
 For a random variable $X$, we use $F_X$ for its distribution function, and let $F^{-1}_X(t)$ be its left $t$-quantile, given by $F^{-1}_X(t)=\inf \{x\in \R: \p(X\le x)\ge t\}$ for $t\in (0,1)$.   Let 
$\Delta_{n}$ denote the standard $n$-simplex, that is,
$  \Delta_{n}=\left\{(x_1,\dots,x_n) \in[0,1]^{n}: x_{1} +\dots+x_{n}=1\right\}.$

Let  
$\M^q$ be the distributions of elements of $L^q$ for $q\in [1,\infty]$.
Let
$L^q_+=\{X\in L^q:X\ge 0,~ X\ne 0\}$
and 
$\M^q_+$ be the distributions of elements of $L^q_+$.  
For a set $\X$ of random variables and a   mapping $\rho:\X \to \R^k$ that is law invariant (i.e., $\rho(X)=\rho(Y)$ if $X$ and $Y$ are identically distributed),
we abuse the notation by treating it also as a mapping from a set $\mathcal M_\rho$ of distributions to $\R^k$,  that is, $\rho(F_X)=\rho(X)$. 
In this case, $\mathcal M_\rho$ is the set of   distributions of  random variables in $\X $.
Since we focus on law-invariant mappings throughout the paper,
this abuse of notation will be helpful in many places.  

The Gini deviation (GD) and the Gini coefficient  (GC) are two very important indices of dispersion, deviation, and economic inequality, and they have wide applications in finance,  economics and machine learning.
The Gini deviation $\GD:L^1\to \R$ is defined as
\begin{align}\label{eq:0}
\GD(X) = \frac 12\E[|X-X'|],
\end{align}
where $X'$ is an iid copy of $X$. 
Alternatively, we can represent $\GD$  in terms of quantile functions or distribution functions, via (see e.g., \citet[Example 1]{WWW20})
$$
\GD(X) = \int_0^1 F_X^{-1}(t) (2t-1)\d t  =\int_\R F_X(x) (1-F_X(x))\d x. 
$$
A  (law-invariant) signed Choquet integral, also called a distortion riskmetric,  is a mapping of the form
 $$
\rho_h:
X\mapsto  \int_{0}^\infty h(\p(X>x))\d x +\int_{-\infty}^0 (h(\p(X>x))-h(1) )\d x,
 $$
 where    $h:[0,1]\to \R$ is a function of bounded variation satisfying $h(0)=0$, called a  distortion function.
 The functional $\GD$ belongs to this class and has a distortion function $h:[0,1]\to \R$ given by $h(t)=t(1-t)$.

The Gini coefficient $\GC: L^1_+\to [0,1]$ is defined as 
$$
\GC (X) =\frac{\GD(X)}{\E[X]} = \frac{\int_0^\infty F_X(t) (1-F_X(x))\d x }{\int_0^\infty (1-F_X(x))\d x}.
$$
The range of $\GC$ is $[0,1)$.
An alternative way of defining $\GC$ is through the Lorenz curve (\cite{G71}),
$$L_{F_X}(p) = \frac{\int_0^p F_X^{-1}(t) \d t}{\int_0^1 F_X^{-1}(t)\d t},~~p\in [0,1],$$
and GC is twice the area between $L_{F_X}$ and the identity on $[0,1]$,
$$
\GC(X)= 2\int _0^1 (p - L_{F_X}(p)) \d p = 1-2\int_0^1 L_{F_X}(p)\d p.
$$

We generally refer to both GD and GC as Gini indices. Since GC and GD are connected by the simple relation
$\GC=\GD/\E$,  our axiomatic framework in the next section will mainly focus on properties satisfied by GD.  

\section{An axiomatic framework for higher-order Gini indices}\label{sec:3}

\subsection{Axioms of the Gini indices}

Our aim is to generalize GD and GC to a family of flexible indices, while keeping the nice properties of GD and GC intact.
For this purpose, we will propose several axioms satisfied by GD, and then look for functionals that also satisfy these axioms. We assume  that the domain of the indices is $\X=L^q$ for some $q\ge 1$, and denote by $\X_+=L^q_+$.

The first observation, clearly from \eqref{eq:0},  is that GD 
can be written as the expected value of a function of two iid copies $X_1,X_2$ from $X$. That is, there exists a function $f:\R^2\to \R$ such that 
$$
\GD(X) = \E[f(X_1,X_2)]. 
$$
A natural generalizing of  this property is to use $n\ge 3$ iid copies, which leads to the following axiom of Gini indices.
We use $\rho:\X\to \R$ for a general functional as a candidate for a generalization of GD. 

\begin{enumerate}
\item[{[A1]}] Sample representability: There exists $f:\R^n\to \R$ for some $n\in \N$ such that $\rho(X)=\E[f(X_1,\dots,X_n)]$ for all $X\in \X$, where $X_1,\dots,X_n$ are iid copies of $X$. 
\end{enumerate}

Besides [A1], GD satisfies many other natural axioms, which we list below.  

{%
\renewcommand{\labelenumi}{[A\arabic{enumi}]}
\begin{enumerate} \setcounter{enumi}{1}
\item Symmetry: $\rho(X)=\rho(-X)$ for $X \in\X$.
\item Comonotonic additivity: $\rho(X+Y)=\rho(X)+\rho(Y)$ if $X$ and $Y$ are comonotonic.\footnote{Random variables $X$ and $Y$ are comonotonic if there exists $\Omega_0 \in \mathcal{A}$ with $\mathbb{P}\left(\Omega_0\right)=1$ such that for all  $\omega, \omega^{\prime} \in \Omega_0$,
$
\left(X(\omega)-X\left(\omega^{\prime}\right)\right)\left(Y(\omega)-Y\left(\omega^{\prime}\right)\right) \geq 0.
$}

\item Uniform norm continuity: For any $\epsilon>0$, there exists $\delta>0$ such that $|\rho(X)-\rho(Y)|\le \epsilon $  whenever $\esssup |X-Y|\le \delta$, where $\esssup$ means the essential supremum.
\item Nonnegativity: $\rho(X)\ge 0$ for $X\in \X$ and $\rho(X)=0 $ if and only if $X$ is a constant.
\item Location  invariance :  $\rho(X+c)=\rho(X)$ for all  $c\in \R$ and $X\in \X$.
\item Positive homogeneity:  $\rho(\lambda X)=\lambda \rho(X)$ for all $\lambda \in (0, \infty)$ and $X\in \X$.

\item Convexity: $\rho(\lambda X+(1-\lambda)Y)\le \lambda\rho(X)+(1-\lambda)\rho(Y)$ for all $\lambda\in[0,1]$ and  $X,Y\in \X$. 
\item Subadditivity:  $\rho(X+Y)\le \rho(X)+\rho(Y)$ for all $X,Y\in \X$.
\item Convex-order consistency: $\rho(X) \leq \rho(Y)$ for all $X,Y\in \X$ whenever  $X \leq_{\rm cx} Y$.\footnote{A random variable $X$ is said to be smaller than a random variable $Y$ in convex order, denoted by $X \leq_{\rm cx} Y$, if $\mathbb{E}[\phi(X)] \leq \mathbb{E}[\phi(Y)]$ for all convex $\phi: \mathbb{R} \rightarrow \mathbb{R}$, provided that both expectations exist.}
\item Mixture-concavity:      $ \rho$ is concave as a mapping from $\M_\rho$ to $\R$, that is, $\rho(\lambda F+(1-\lambda)G)\ge \lambda\rho(F)+(1-\lambda)\rho(G)$ for all  $F,G\in \M_\rho$ and $\lambda\in[0,1]$. 
\item Normalization: $\{\rho(X)/\E[X]:X\in \X_+\}=[0,1)$.

\end{enumerate}} 
Axioms [A1]--[A12] are all satisfied by GD, and we will  refer them  to as the \emph{Gini axioms}. Some axioms imply the others. For instance,   positive homogeneity [A7] and convexity [A8]   imply subadditivity [A9].

Note that [A12] is made so that $\rho/\E$ can be readily used as the corresponding relative index playing the role of GC.
Therefore, the Gini axioms are made for the pair $(\rho,\rho/\E)$ in place of $(\GD,\GC)$, although the mathematical statements  for $\rho$ are simpler to present, except for [A12]. The economic interpretation of some of these axioms will be discussed in Section \ref{sec:property}.


\subsection{Axiomatic characterization}
 Our aim is to generalize the Gini deviation  and the Gini coefficient to higher order.
In what follows, always assume $n\ge 2$. 
To establish a representation with as few axioms as possible, we will only assume [A1]--[A4] in the main direction of the result.

\begin{theorem}\label{th:chara}
For a mapping $\rho: \X \to \R$,   it satisfies {\rm [A1]--[A4]} if and only if there exist an integer $n\in \N$ 
and $(a_1,\dots,a_n)\in  \R^n$ such that 
\begin{align}\label{eq-affineGD}
\rho(X)=\sum_{i=1}^n a_i {\rm GD}_i(X),~~ X\in \X,
\end{align}
where $\GD_1(X)=\GD(X)$, and for $i=2,3,\dots,n$, 
\begin{align}\label{eq:0c}
\GD_i(X) = \frac 1 i \E \left[\max\{X_1,\dots,X_i\}-\min\{X_1,\dots,X_i\}\right],
\end{align}
with  $X_1,\dots,X_n$ being iid copies of $X$.
Moreover, if $(a_1,\dots,a_n)\in \Delta_{n}$, then 
$\rho(X)$ in \eqref{eq-affineGD}  satisfies all of the Gini axioms {\rm [A1]--[A12]}. 
\end{theorem}

The most important part of the proof of 
Theorem \ref{th:chara} is 
 the necessity of \eqref{eq-affineGD}. This involves two steps. The first step establishes a  representation via signed Choquet integrals. This  follows from results in  \cite{WWW20a}, which uses [A3], [A4], and law invariance (implied by [A1]) to establish the representation. The second step identifies symmetric ([A2]) signed Choquet integrals that can be represented by expected values of functions of iid observations ([A1]), a new property to the literature that requires sophisticated analysis.

Inspired by the representation \eqref{eq-affineGD} in Theorem \ref{th:chara},  we define the  \emph{$n$-th order Gini deviation} for $n\ge 2$ as the functional $\GD_n$ in \eqref{eq:0c}, and without loss of generality its domain is chosen as $L^1$. That is,
\begin{align*} 
\GD_n(X) = \frac 1 n \E \left[\max\{X_1,\dots,X_n\}-\min\{X_1,\dots,X_n\}\right], ~~~~ X \in L^1. 
\end{align*}
Note that \eqref{eq:0} can be rewritten as
 $$
\GD(X) = \frac 12\E \left[\max\{X,X'\}-\min\{X,X'\}\right],
  $$
and therefore, $\GD_2=\GD$. 
Following this,   the
 \emph{$n$-th order Gini coefficient} $\GC_n$ for $n\ge 2$, is defined by, as in \cite{GROG24}, 
$$
\GC_n (X)= \frac{\GD_n(X)}{\E[X]}, ~~~~ X \in L^1_+. 
$$
Note that $\GC_n$ has a natural range of $[0,1)$ due to [A12], and this fact has  been observed by \cite{GROG24}. 
The following example explains the extreme distributions, and also justifies that
  $\GD_n$ and $\GC_n$ bear the same interpretation as $\GD$ and $\GC$ when it comes to measuring wealth  or income inequality, as the worst-case scenario is attained asymptotically by the distribution in which a tiny ($\epsilon$ in the example) proportion of the population has all the wealth or income.

\begin{example}\label{ex:1}
Suppose that $X_\epsilon$ follows a Bernoulli distribution with mean $\epsilon>0$.  
From the definition of $\GD_n$, we can compute 
$$\GC_n(X_\epsilon) =\frac{\GD_n(X_\epsilon)}{\E[X_\epsilon]}=\frac 1 \epsilon \int_{1-\epsilon}^1  ( t^{n-1} -(1-t)^{n-1} )\d t
=\frac{1}{n\epsilon} \left( 1-\epsilon^n -(1-\epsilon)^n \right),
$$and by sending $\epsilon\downarrow0$,
 we get
$$\GC_n(X_\epsilon)  
=\frac{1}{n\epsilon} \left(  n \epsilon +O(\epsilon^2) \right)\to 1.
$$ 
\end{example}

Some members in the family of high-order Gini deviations coincide. 
In addition to $\GD_2=\GD$, we can also check   $\GD_2=\GD_3$ by the equality
\begin{align*}
 \max\{X_1,X_2,X_3\}-\min\{X_1,X_2,X_3\} = \frac{1}{2} (|X_1-X_2|+|X_1-X_3|+|X_2-X_3|),
  \end{align*}
which implies 
  \begin{align*}
\GD_3(X) =  \frac{1}{6} \E[|X_1-X_2|+|X_1-X_3|+|X_2-X_3|] = \frac 12 \E[|X_1-X_2|]=\GD_2(X).
  \end{align*}
Nevertheless, $\GD_4\ne \GD_3$ in general, and $\GD_n$ decreases in $n$ (see Appendix \ref{sec:monotone}). 

Inspired by the axiomatic characterization of indices satisfying the Gini axioms in Theorem \ref{th:chara}, 
we focus on the families $(\GD_n)_{n\ge2}$ and $(\GC_n)_{n\ge 2}$ of higher-order Gini indices in the rest of the paper. 
In particular, $\GC_n$ for some different choices of $n\ge 2$ will be our main tool for measuring inequality as a complement to the classic $\GC$, 
while we should keep in mind that their convex combinations also satisfy the Gini axioms and can be used for the same purpose.

\begin{remark}
In Theorem \ref{th:chara}, we have shown that if $(a_1,\dots,a_n)$ is in $\Delta_n$, then all the Gini axioms are satisfied. One may wonder whether the reverse holds true, that is, whether for $\rho$ in \eqref{eq-affineGD},  [A1]--[A12] jointly imply $(a_1,\dots,a_n)\in \Delta_n$.
It turns out that this is not true, even removing the redundant terms (e.g., $\GD_1=\GD_2=\GD_3$). For an example, define $\rho(X)=2 {\rm GD}_2(X)-{\rm GD}_4(X)$. By Proposition \ref{th:basic}, we know that $\rho$ is a signed Choquet integral whose distortion function is given by \begin{align*}
h(t)= (1-t^2-(1-t)^2)-\frac{1}{4}(1-t^4-(1-t)^4).
\end{align*}
One can check that $h$ is concave on $[0,1]$, and thus, $\rho$ is convex. Using this fact, we can check all of [A5]--[A12] by standard computations.
\end{remark}


\section{Properties of the higher-order Gini indices}\label{sec:property}

\subsection{Economic properties of higher-order GC}

The next theorem gives some useful and economically interpretable properties of $\GC_n$, similar to the properties of $\GD_n$ in Theorem \ref{th:chara}, but with different forms. 
  
\begin{theorem}
\label{prop:GD-p}
    The mapping $\GC_n:L^1_+\to \R$ 
    satisfies 
    \begin{enumerate}[(i)]  
\item Scale invariance:  $\GC_n(\lambda X)= \GC_n(X)$ for all $\lambda \in (0, \infty)$ and $X\in L^1_+$; 
\item Convex-order consistency: $\GC_n (X) \leq \GC_n (Y)$ for all $X,Y\in L^1_+$ whenever  $X \leq_{\rm cx} Y$; 
\item Quasi-convexity: $\GC_n(\lambda X+(1-\lambda)Y)\le \max\{  \GC_n(X), \GC_n(Y)\}$ for all $\lambda\in[0,1]$ and  $X,Y\in L^1_+$;   
\item Mixture-quasi-concavity:   $ \GC_n$ is quasi-concave as a mapping from $\M^1_+$ to $\R$, that is, $\GC_n(\lambda F+(1-\lambda)G)\ge \min\{ \GC_n(F), \GC_n(G)\}$ for all  $F,G\in \M^1_+$ and $\lambda\in[0,1]$. 
\end{enumerate}
\end{theorem}



The four properties of $\GC_n$ in Theorem \ref{prop:GD-p} have natural economic interpretations as measures of inequality, similarly to the case of $\GC$. To explain these interpretations, we use measuring   income inequality as the main example. 
\begin{enumerate}
    \item Scale invariance means that a rescaling of the income distribution, e.g., multiplying by a currency exchange rate, does not affect the measurement. Therefore, $\GC_n$ can be evaluated in either the local currency or a standard one, such as the US dollars.
    \item Convex-order consistency of $\GC_n$ means that a more spread-out distribution has a larger $\GC_n$, confirming the role of $\GC_n$ as a measure of distributional dispersion. 
    In particular, if the income process (with respect to time) in a country forms a martingale (i.e., the expected income in the next year of an individual is equal to their income this year), then $\GC_n$ will increase in time.  
    \item 
Quasi-convexity of $\GC_n$ means that if we aggregate income $X$ and income $Y$ from two populations of equal size in pair, for instance, forming households from pairs of individuals, 
then the overall value of $\GC_n $ is smaller than the maximum of $\GC_n(X)$ and $\GC_n(Y)$.
This means that aggregating groups of individuals with similar level of inequality generally reduces inequality. 

\item   Mixture-quasi-concavity of $\GC_n$ means that if we merge two populations with different income distributions $X$ and $Y$ to make a bigger population (for instance, forming  a larger region from two small  regions), then the overall value of $\GC_n $ is larger than the minimum of $\GC_n(X)$ and $\GC_n(Y)$. 
This suggests that a larger country likely has a larger $\GC_n$.
This is a natural condition for measuring inequality. For example, two populations that are perfectly equal among themselves (e.g., one rich group and one poor group)
can exhibit significant inequality when assessed as a combined population. 
\end{enumerate}

Convex-order consistency is also used to define strong risk aversion, that is, aversion to mean-preserving spreads in the sense of \cite{RS70}. 
Strong risk aversion of preference models is widely studied in decision theory. For instance, \cite{SZ08} provided a characterization of strong risk aversion within the framework of cumulative prospect theory.
 
The above point 3 (quasi-convexity) and point 4 (mixture-quasi-concavity) should not be seen as conflicting, as summing two income random variables and 
merging two populations are very different operations and have opposite effects. Indeed, there is a  fundamental conflict between convexity and mixture-convexity (see \citet[Proposition 6]{WW25}), whereas convexity and mixture-concavity are equivalent for Signed Choquet integrals (\citet[Theorem 3]{WWW20}).

\subsection{Representation as signed Choquet integrals}
\label{sec:GDrepresentaton}

In the next proposition, we represent $\GD_n$ as a signed Choquet integral. This is useful in proving several properties of $\GD_n$ in Theorem \ref{th:chara}.  In what follows, let $U_X$ be
a uniform random variable such that $F_X^{-1}(U_X)=X$ almost surely, and its existence is guaranteed by \citet[Lemma A.32]{FS16}. 

\begin{proposition}\label{th:basic}
The mapping $\GD_n:L^1\to \R$ satisfies the following  properties.
\begin{enumerate}[(i)]
\item It has the quantile representation
\begin{equation}
\GD_n(X) = \int_0^1 F_X^{-1}(t) ( t^{n-1} -(1-t)^{n-1} )\d t,~~~~X\in L^1. \label{eq:choquet-1}
\end{equation}
\item It is  a signed Choquet integral with a concave distortion function
\begin{align}\label{eq:h}
h_n(t) = \frac{1}n \left( 1-t^n - (1-t)^n\right),~~~~t\in[0,1].
\end{align}
%
\end{enumerate} 
\end{proposition}



Figure \ref{fig:h_n} depicts $h_n(t)$ in \eqref{eq:h} and its derivative $h_n'(t)$ for various values of $n$. As $n$ increases, $h_n(t)$ becomes flatter over the interval $[0,1]$, exhibiting increasingly sharp transitions near the endpoints $t=0$ and $t=1$. Note that $h_n'(t)$ serves as the “probability weight” assigned to $t$ in the calculation of the Choquet expectation; see, e.g.  \cite{Q82} and \cite{GS89}. From this perspective,  $\GD_n$ also reflects an increasing sensitivity to tail behavior and concentration effects.

\begin{figure}[t!]
\centering
 \includegraphics[width=16cm]{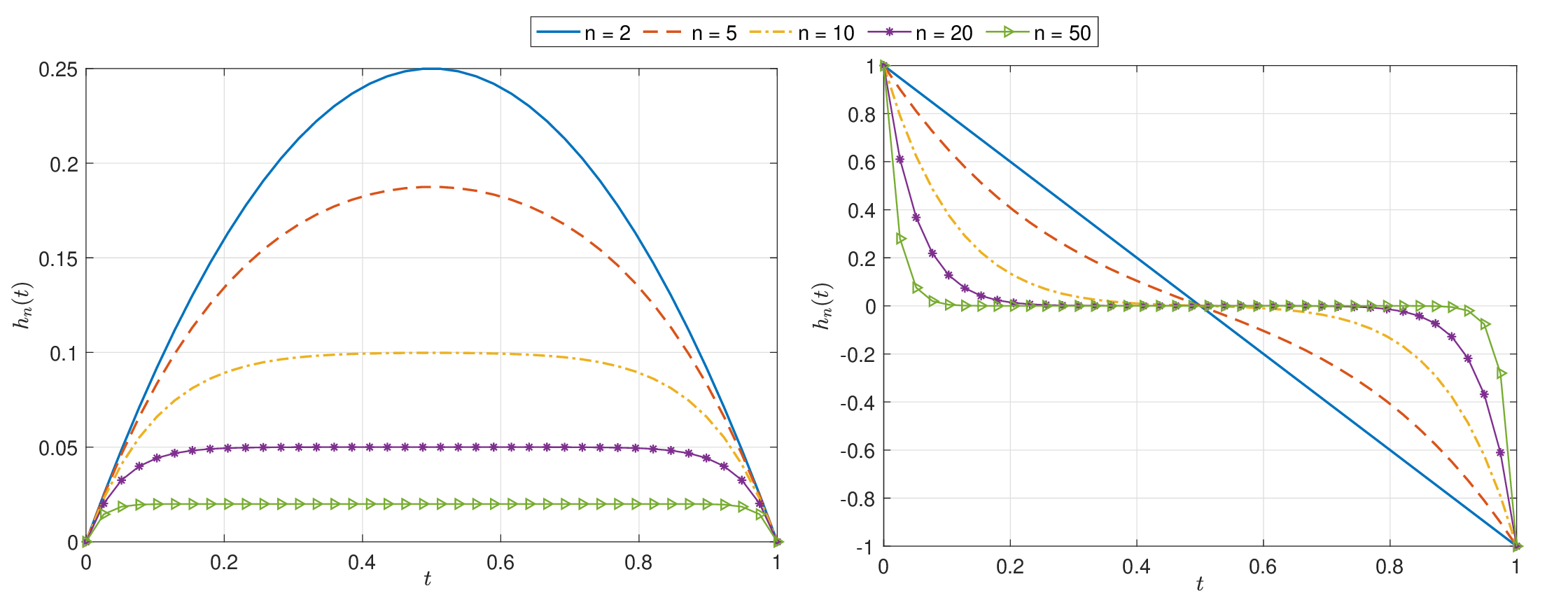}
 \captionsetup{font=small}
 \caption{\small The distortion function $h_n(t) = \left(1 - t^n - (1-t)^n\right)/n$ (left panel)  and  its derivative $h_n'(t)=(1-t)^{n-1}-t^{n-1}$ (right panel) for $n = 2, 5, 10, 20, 50$.
}\label{fig:h_n}

\end{figure}


\begin{remark}  Let $\rho_n$ be the \emph{power distortion risk measure} defined by
  $$
  \rho_n(X) = \int_0^1 n t^{n-1}F^{-1}_X(t) \d t,~~~X\in L^1.
  $$
  It is well known that $\rho_n$ is a coherent risk measure (\cite{ADEH99}). 
  The functional $\GD_n$  is connected to  $\rho_n$ via 
  $$
  \GD_n(X)= \frac {\rho_n(X)+\rho_n(-X)}n.
  $$
\end{remark}

\begin{remark}
    From \eqref{eq:choquet-1}, we can easily generalize the definition of $\mathrm{GD}_n$ to real numbers $n> 1$, and these indices are also signed Choquet integrals. Similarly, one can define $\mathrm{GC}_n$ for all real numbers $n>1$. 
    Nevertheless, for a non-integer $n$, sample representability [A1] is no longer satisfied by $\mathrm{GD}_n$ in  \eqref{eq:choquet-1}.  For this reason, we focus integer values of $n$. 
\end{remark}

\subsection{Comparing   Gini deviations and the standard deviation} \label{sec:4}

The inequality $
0\le \GD(X)/\mathrm{SD}(X) \le 3^{-1/2}
$ for all nonconstant  $X\in L^2$ is known as Glasser's inequality (\cite{G62}), where $\mathrm{SD}(X)$ denotes the standard deviation of $X$. 
For the general $\GD_n$, we establish the following bounds. The proof follows a similar approach to Theorem 5 of \cite{PWW25}.

\begin{proposition}\label{prop:comSD}
For $n\geq2$ and nonconstant $X\in L^2$, we have    \begin{equation}
    \label{eq:prop-SD}
0\leq\frac{\GD_n(X)}{\mathrm{SD}(X)}\leq\sqrt{\frac{2}{2n-1}-\frac{2((n-1)!)^2}{(2n-1)!}},
\end{equation}
and these bounds are sharp.
\end{proposition}

Note that the right-hand side of \eqref{eq:prop-SD} decreases in $n$ and it is very close to $n^{-1/2}$ for $n$ large. Its maximum value is $3^{-1/2}$ when $n=2,3$.


Next, we consider the bounds on 
$\GD_n(X)/\mathrm{GD}_m(X)$.  In fact, we can derive  a more general result for  bounds on the ratio of any two deviation measures defined via the Choquet integrals.   Let  $$\mathcal{H}=\{h: h \mbox{~maps $[0,1]$ to $\R$ with $h(0)=h(1)=0$ and $h(t)>0$ for $t\in(0,1)$}\}.$$
 For  $h,g\in\mathcal H$  and a nonconstant $X$ such that $\rho_{h}(X)$ and $\rho_{g}(X)$ are finite,
 we have \begin{equation}\label{eq:r}\begin{aligned} \frac{\rho_{h}(X)}{\rho_{g}(X)}= 
 \frac{\int_{-\infty}^\infty h(\p(X>x))\d x}{\int_{-\infty}^\infty g(\p(X>x))\d x}
 \in\left[\inf_{t\in(0,1)} \frac{h(t)}{g(t)}, \sup_{t\in(0,1)} \frac{h(t)}{g(t)}\right].\end{aligned}\end{equation} 
The proof of \eqref{eq:r} follows directly from the definition of Choquet integrals. Moreover, we can see that these bounds are sharp as the value $h(t)/g(t)$ for $t\in (0,1)$ is attainable by a  Bernoulli random variable with mean $t$.
A specialization of \eqref{eq:r} to $\GD_n$ leads to the following bounds. 

\begin{proposition}\label{prop:3}
Let $L^1_*$ be the set of nonconstant elements of $L^1$. For $2\leq m\leq n$, we have    \begin{equation}\label{eq:B1}
\min_{X\in L^1_*} \frac{\GD_n(X)}{\mathrm{GD}_m(X)} = 
\frac{m(1-2^{1-n}) }{n(1-2^{1-m}) };~~~~~
\sup_{X\in L^1_*}
\frac{\GD_n(X)}{\mathrm{GD}_m(X)}=1.\end{equation} 
\end{proposition}

The same bounds in Proposition \ref{prop:3} also apply to 
$\GC_n(X)/\GC_m(X)$
for  $X\in L^1_+$. 
In particular, these bounds imply
$1\le \GC_3(X)/\GC_2(X) \le 1$, thus $\GC_3=\GC_2$,
and
$7/8\le \GC_4(X)/\GC_3(X) \le 1$.

   \section{Elicitability and multi-observation elicitability}\label{sec:5}
 Statistical  elicitability is an important
 property for statistical functionals, popular in forecasting, risk management, and machine learning;  see \cite{G11} for a thorough study.   It 
 refers to the existence of a scoring function for the forecasted value of a risk functional and realized value of future observations, so that the mean of the scoring function attains its minimum value if and only if the value of the risk functional is truly forecasted.  Comparative backtests, for which elicitability is a necessary condition, are discussed by \cite{NZ17} as an alternative to the traditional backtests.

In this section, we mainly treat $\rho$ as mappings from $\M^q$ to $\R$ for some $q\ge 1$, instead of $L^p$ to $\R$ for some $q\ge 1$. 
The notation ``$f:  A \rightrightarrows  B$" means 
  that $f$ maps each element of $A $ to a subset of  $B$.
  Moreover, we identify a singleton in $\R^k$ with its element. For instance, if $\rho:  \mathcal{M} \rightrightarrows  \mathbb{R}^k$ takes value only in singletons, then we treat $\rho : \mathcal{M} \to  \mathbb{R}^k$.

\begin{definition}\label{def:1} For $\mathcal M\subseteq \mathcal M_\rho$, the mapping $\rho: \mathcal{M}_\rho \rightrightarrows  \mathbb{R}^k$ is 
\emph{$\mathcal{M}$-elicitable} if there exists a   function $S: \mathbb{R}^{k+1} \rightarrow \mathbb{R}$ such that for all $F \in \mathcal{M}$,
$$
\rho(F)=\argmin_{\mathbf x\in \R^k}  \int_{-\infty}^{\infty} S(\mathbf{x}, y) \mathrm{d} F(y) .
$$
We omit ``$\mathcal M$" in ``$\mathcal M$-elicitability" when $\M$ is the domain $\mathcal M_\rho$ of the mapping $\rho$.
The function $S$ is called a score function.
\end{definition}
  It is easy to check that none of GD or GC is elicitable. In fact, as shown by \cite{WW20}, all deviation measures (including GD) are not elicitable. 
  
  Inspired by the definition of GD, $\GD_n$ is connected to another version of elicitability. This idea of multi-observation elicitability appeared  first in \cite{CFMW17} and it is further studied by \cite{FMMW17}.

\begin{definition}\label{def:2}
For $\mathcal M\subseteq \mathcal M_\rho$, a mapping $\rho: \mathcal{M}_\rho \rightrightarrows \mathbb{R}^k$ is \emph{$n$-observation $\mathcal{M}$-elicitable} if there exists a function $S: \mathbb{R}^{k+n} \rightarrow \mathbb{R}$ such that
$$
\rho(F)=\argmin_{\mathbf x\in \R^k}\int_{\mathbb{R}^n} S(\mathbf{x}, \mathbf{y}) \mathrm{d} F^n(\mathbf{y}), \quad F \in \mathcal{M},
$$
where $F^n\left(y_1, \ldots, y_n\right)=\prod_{j=1}^n F\left(y_j\right)$.
\end{definition}

Two prominent examples of 2-observation elicitable functionals 
are the variance (Example \ref{ex:variance} below) and the Gini deviation (Theorem \ref{thm:2} below).

\begin{example}
\label{ex:variance}
Similar to the Gini deviation in \eqref{eq:0}, the variance functional can be written as
$$
\var(X) = \frac 12 \E[(X-X')^2],
$$
where $X'$ is an iid copy of $X$.
The variance is $2$-observation $\M^4$-elicitable with the  score function  (e.g., \cite{CFMW17})
$$
S(x, y_1,y_2) = (2x-(y_1-y_2)^2)^2.
$$
This can be checked directly from
$$
 \argmin_{  x\in \R  }\E\left[(2x-(X-X')^2)^2 \right] 
 = \frac 12 \E[ (X-X')^2] = \var(X).
$$
\end{example}
Clearly, $1$-observation elicitability coincides with elicitability in Definition \ref{def:1}.
Moreover, the property of $n$-observation elicitability gets weaker as $n$ grows. 
\begin{proposition}\label{prop:2}
If $m\ge n$, then $n$-observation elicitability implies $m$-observation elicitability.
\end{proposition}
 This proposition follows  directly  by noticing that the function $S_m$ given by $S_m(x,y_1,\dots,y_m)=S_n(x,y_1,\dots,y_n)$, where $S_n$ is the $n$-observation score function, is indeed an $m$-observation score function.

Next, we establish two score functions for $\GD_n$.

\begin{theorem}\label{thm:2}
The mapping $\mathrm{GD}_n$ is $n$-observation $\mathcal{M}^2$-elicitable with  the score function $$S\left(x, y_1, \dots, y_n\right)=\left(n x- \max\{y_1,\dots, y_n\} + \min\{y_1,\dots, y_n\}  \right)^2.$$ 
Moreover, $\mathrm{GD}_n$ is $n$-observation $\mathcal{M}^1$-elicitable with  the score function  $$S^*\left(x, y_1, \dots, y_n\right)= n x^2-2x \left(\max\{y_1,\dots, y_n\}-\min\{y_1,\dots, y_n\} \right).$$ 
 In particular,    $\mathrm{GD}_2$ is $2$-observation $\M^2$-elicitable by $ (x, y_1, y_2 )\mapsto (2 x- |y_1-y_2 | )^2$  
 and $2$-observation  $\M^1$-elicitable by $ (x, y_1, y_2 )\mapsto x^2-x|y_1-y_2|$.
\end{theorem}

 
Similarly to the classic notion of elicitability, the score function 
for an $n$-observation elicitable functional 
is not unique in general, as suggested by Theorem \ref{thm:2}.  

The next result shows $n$-observation elicitability of signed Choquet integrals with  a polynomial distortion function.
\begin{theorem}\label{th-n-elicitability-distortion}
 A signed Choquet integral, with   distortion function  $h$ that is a polynomial  of degree $n$ or less, 
is $n$-observation $\mathcal{M}^1$-elicitable.
\end{theorem}


 When $n$ is an odd number,  we can check that the $n$-th degree term  $t^n$ in the distortion function $h_n$ in \eqref{eq:h}
get canceled. Therefore, in this case, by using Theorem \ref{th-n-elicitability-distortion}, $\GD_n$ is $(n-1)$-observation elicitable, which is a stronger condition than $n$-observation elicitability, as shown in Proposition \ref{prop:2}. 
Note that the score function for $n-1$ observations is different from the one obtained in Theorem \ref{thm:2} for $n$ observations.

The functional $\GC_n$ is also $n$-observation elicitable, as the ratio of expectations, following the same idea as in  Theorem 3.2 of \cite{G11}. 
\begin{proposition}\label{prop:elicitability-GC}
The mapping $\mathrm{GC}_n$ is $n$-observation $\mathcal{M}^1$-elicitable with  the score function $$S\left(x, y_1, \dots, y_n\right)=x^2 y_1- \frac{2x}{n}(\max\{y_1,\dots, y_n\} - \min\{y_1,\dots, y_n\}).$$  In particular,    $\mathrm{GC}_2$ is $2$-observation $\M^1$-elicitable by $S (x, y_1, y_2 )=  x^2y_1- x|y_1-y_2 | $.
\end{proposition}

The elicitability results allow for obtaining $\GD_n $ and $\GC_n$ through empirical risk minimization. In the next section, we will study estimators of $\GD_n$ and $\GC_n$ through their Choquet integral representation.


\section{Estimation and inference}\label{sec:6}

In this section, we study empirical estimators of $\GD_n$ and $\GC_n$ and establish their consistency and asymptotic normality, which can be used for point estimation and hypothesis testing.

Let $X_1, \ldots, X_N$ be a sample of independent observations, and denote by $X_{(1)} \leq \cdots \leq X_{(N)}$ their order statistics. The empirical distribution function is given by  $\widehat{F}_N(x) = \frac{1}{N} \sum_{j=1}^N \mathbf{1}_{\{X_j \leq x\}},$  and its associated quantile function satisfies  $\widehat{F}_N^{-1}(i/N) = X_{(i)}$ for $i = 1, \ldots, N$. 
As a standard approach for estimating  L-statistics and  distortion risk measures   (e.g., \cite{HR09} and \cite{MFE15}), 
the empirical estimator of  $\GD_n(X)$ is then defined as
$$
\widehat{\GD}_n(N) = \frac{1}{N} \sum_{i=1}^N X_{(i)} \left( \left( \frac{i}{N} \right)^{n-1} - \left( 1 - \frac{i}{N} \right)^{n-1} \right),
$$
while the empirical estimator of  $\GC_n(X)$ takes the form

$$
\widehat{\GC}_n(N) = \frac{1}{\overline{X}} \cdot \widehat{\GD}_n(N), \quad \text{where } \overline{X} = \frac{1}{N} \sum_{i=1}^N X_i.
$$
 We make  following standard regularity assumption on the distribution of the random variable $X$.
\begin{assumption}\label{assump:1}
The distribution $F$ of $X \in L^1$ has a  support that is a convex set and has a positive density function $f$ on the support. Denote by $\tilde f=f \circ F^{-1}$. 
\end{assumption}
\begin{theorem}\label{thm:5}Suppose that $X_1,\dots, X_N \in L^1_+$ {are} an iid sample from  $X\in L^1_+$ and Assumption \ref{assump:1} holds. Then,  we have
$
 \widehat{\GD}_n(N) \stackrel{\mathbb{P}}{\rightarrow} \GD_n(X)$ and $
 \widehat{\GC}_n(N) \stackrel{\mathbb{P}}{\rightarrow} \GC_n(X)$  as $n \to \infty.
$
 Moreover, if  $X  \in L^{\gamma}_+$ for some $\gamma>2$, then 
we have   
\begin{align*}
\sqrt{N}\left(\widehat{\GD}_n(N)-\GD_n(X)\right) &\stackrel{\mathrm{d}}{\rightarrow} \mathrm{N}\left(0, \sigma^2_{\GD_n}\right); \\
\sqrt{N}\left(\widehat{\GC}_n(N)-\GC_n(X)\right) &\stackrel{\mathrm{d}}{\rightarrow} \mathrm{N}\left(0, \sigma^2_{\GC_n}\right),
\end{align*}
where
\begin{align}\label{eq:sigma_GD} 
\sigma^2_{\GD_n}&= \int_0^1 \int_0^1 \frac{\phi_n(s) \phi_n(t) (s \wedge t-s t)}{\tilde f(s) \tilde f(t)} \mathrm{d} t \mathrm{d} s, 
\\ \label{eq:sigma_GC} 
\sigma^2_{\GC_n}&= \int_0^1 \int_0^1 \frac{(\phi_n(s)-\GC_n(X) )( \phi_n(t)-\GC_n(X) )(s \wedge t-s t)}{(\E[X])^2\tilde f(s) \tilde f(t)} \mathrm{d} t \mathrm{d} s,
\end{align}
and $ \phi_n(s)=s^{n-1} -(1-s)^{n-1}$ for $s\in [0,1]$.
\end{theorem}

The computational complexity of the empirical estimator $\widehat{\GD}_n(N)$ is dominated by the sorting step required to compute the empirical distribution function. Sorting the sample $X_1, X_2, \dots, X_N$ takes $O(N \log N)$ time using efficient sorting algorithms such as quicksort or mergesort. After sorting, the calculation of the empirical distribution and the estimator itself requires only $O(N)$ operations. Therefore, the overall time complexity of $\widehat{\GD}_n(N)$ is $O(N \log N)$.

In what follows, we present some simulation results based on Theorem \ref{thm:5}. 
Simulation results are presented in the case of log-normal distribution $\mathrm{LN}(0,1)$ and $\text{Pareto}(3,2)$ (with tail index $3$). Let the sample size  $N = 5000$, and we repeat the procedure 2000 times. We assume  $n = 5 $. 
\begin{figure}[t!]
\centering
 \includegraphics[width=16cm]{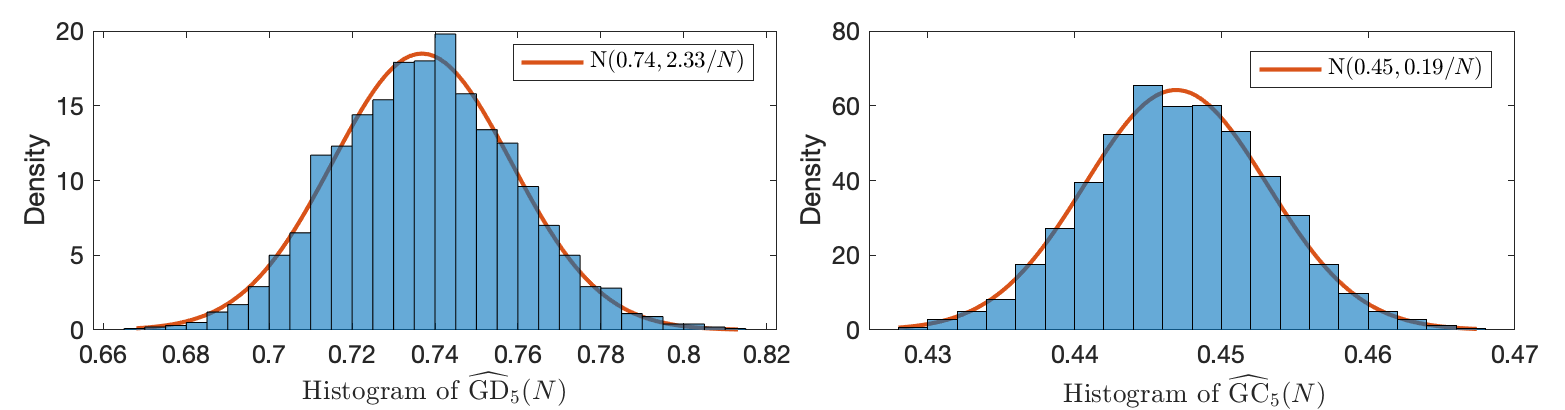}
 \captionsetup{font=small}
 \caption{ \small Values of  $\widehat{\GD}_5(N)$ (left panel)  and $\widehat{\GC}_5(N)$ (right panel) for 
$\mathrm{LN}(0,1)$ }\label{fig:asy_LN}
\end{figure}

\begin{figure}[t!]
\centering
 \includegraphics[width=16cm]{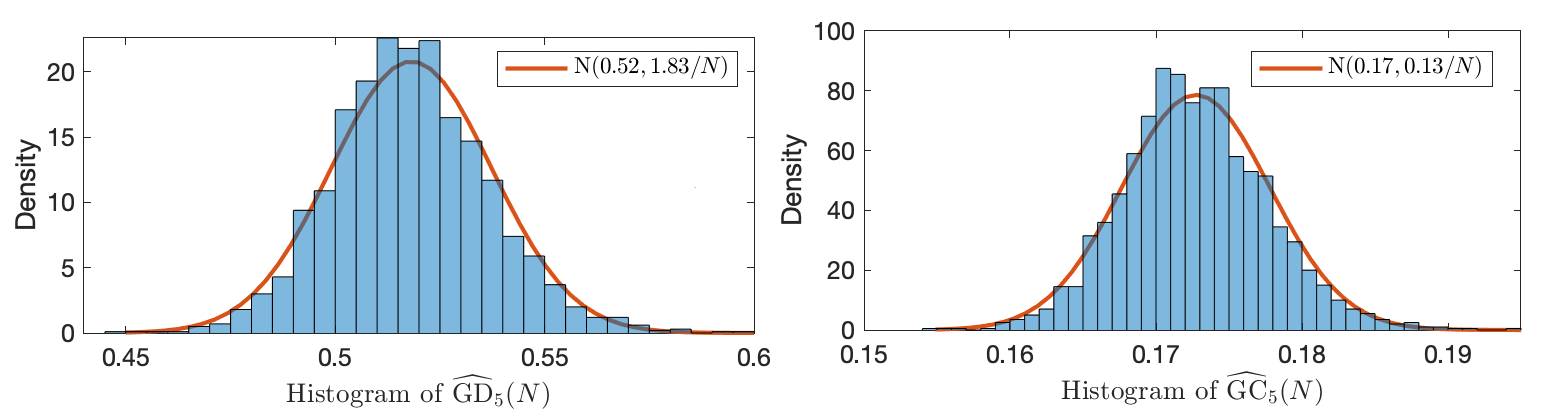}
 \captionsetup{font=small}
 \caption{ \small Values of  $\widehat{\GD}_5(N)$ (left panel)  and $\widehat{\GC}_5(N)$ (right panel) for 
 $\mathrm{Pareto}(3,2)$}\label{fig:asy_pareto}
\end{figure}

In Figure~\ref{fig:asy_LN}, the sample is drawn from the log-normal distribution $\mathrm{LN}(0,1)$. We observe that the empirical estimates of $\GD_5(X)$ align quite well with the density of $\mathrm{N}(0.74,2.33/N)$, while the empirical estimates of $\GC_5(X)$ match the density of $\mathrm{N}(0.45,0.19/N)$. Notably, the asymptotic variance of $\GC_5(X)$ is smaller than that of $\GD_5(X)$.
In Figure~\ref{fig:asy_pareto}, the sample is drawn from the Pareto distribution $\mathrm{Pareto}(3,2)$. We observe that the empirical estimates of $\GD_5(X)$ closely match the density of $\mathrm{N}(0.52,1.83/N)$, and the empirical estimates of $\GC_5(X)$ align with the density of $\mathrm{N}(0.17,0.13/N)$.

\begin{figure}[t!]
\centering
 \includegraphics[width=16cm]{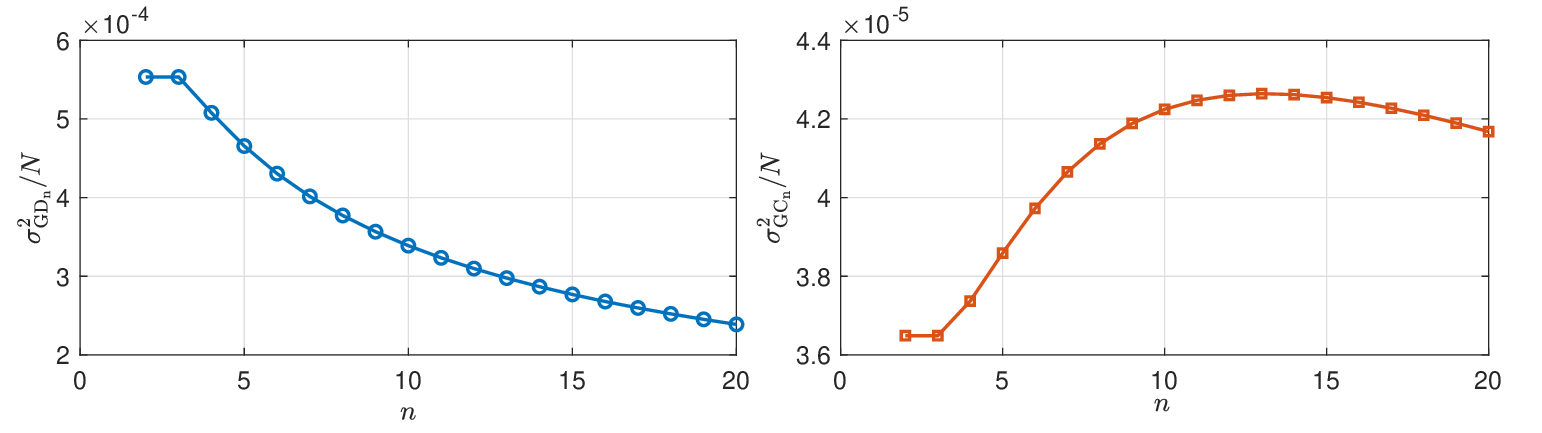}
 \captionsetup{font=small}
 \caption{ \small  Asymptotic variance  for $\GD_n(X)$ (left panel)
 and $\GC_n(X)$ (right panel) for $X\sim\mathrm{LN}(0,1)$}\label{fig:variance1}
\end{figure}
\begin{figure}[t!]
\centering
 \includegraphics[width=16cm]{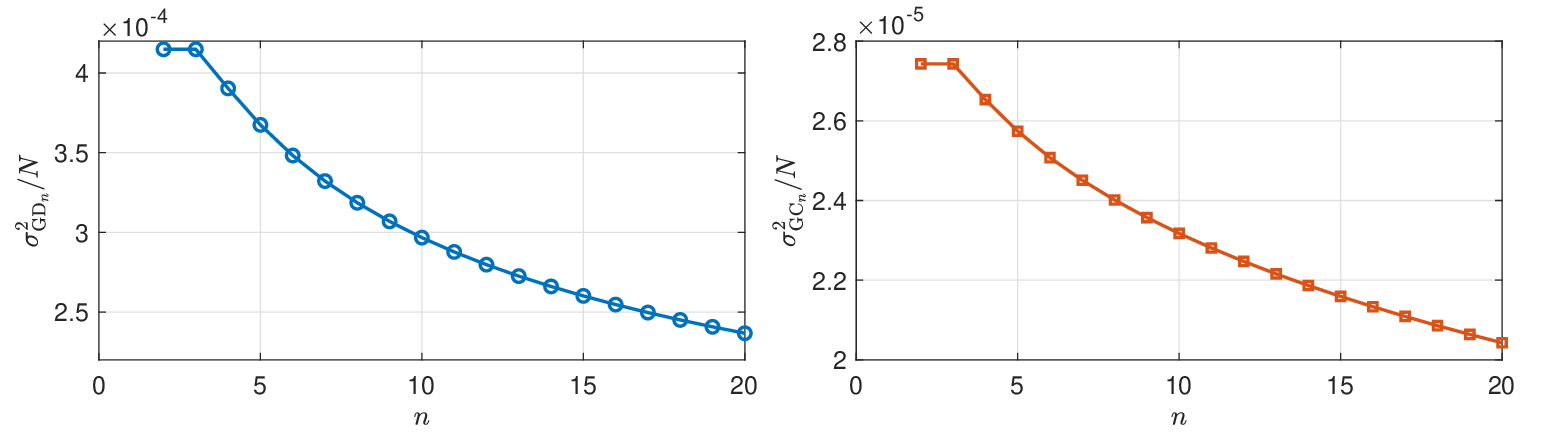}
 \captionsetup{font=small}
 \caption{\small Asymptotic variance  for $\GD_n(X)$ (left panel)
 and $\GC_n(X)$ (right panel) for $X\sim\mathrm{Pareto}(3,2)$}\label{fig:variance2}
\end{figure}

In addition, we illustrate how the value of asymptotic variance varies as $n$ ranges from 2 to 20. As shown in Figures \ref{fig:variance1} and \ref{fig:variance2}, the asymptotic variance of $\GD_n(X)$ decreases as $n$ increases. However, the asymptotic variance of $\GC_n(X)$ does not necessarily exhibit such monotonic behavior. For example, when  $X\sim \mathrm{LN}(0,1)$, the asymptotic variance of $\GC_n(X)$ initially increases and then decreases.

\section{Empirical analysis}\label{sec:8}

In this section, we present an empirical analysis by computing high-order Gini indices, $\mathrm{GC}_n$, for both wealth and income across countries and regions over time, and by comparing these results with the conventional case of $n = 2$. 
As suggested by Theorem \ref{th:chara}, one could also consider linear combinations of $\GC_n$, but this will lead to similar results, which we omit.

Our analysis is based on data from the World Inequality Database (WID), with all monetary figures converted to United States dollars to ensure cross-country comparability.

 \begin{figure}[t!]
\centering
 \includegraphics[width=16cm]{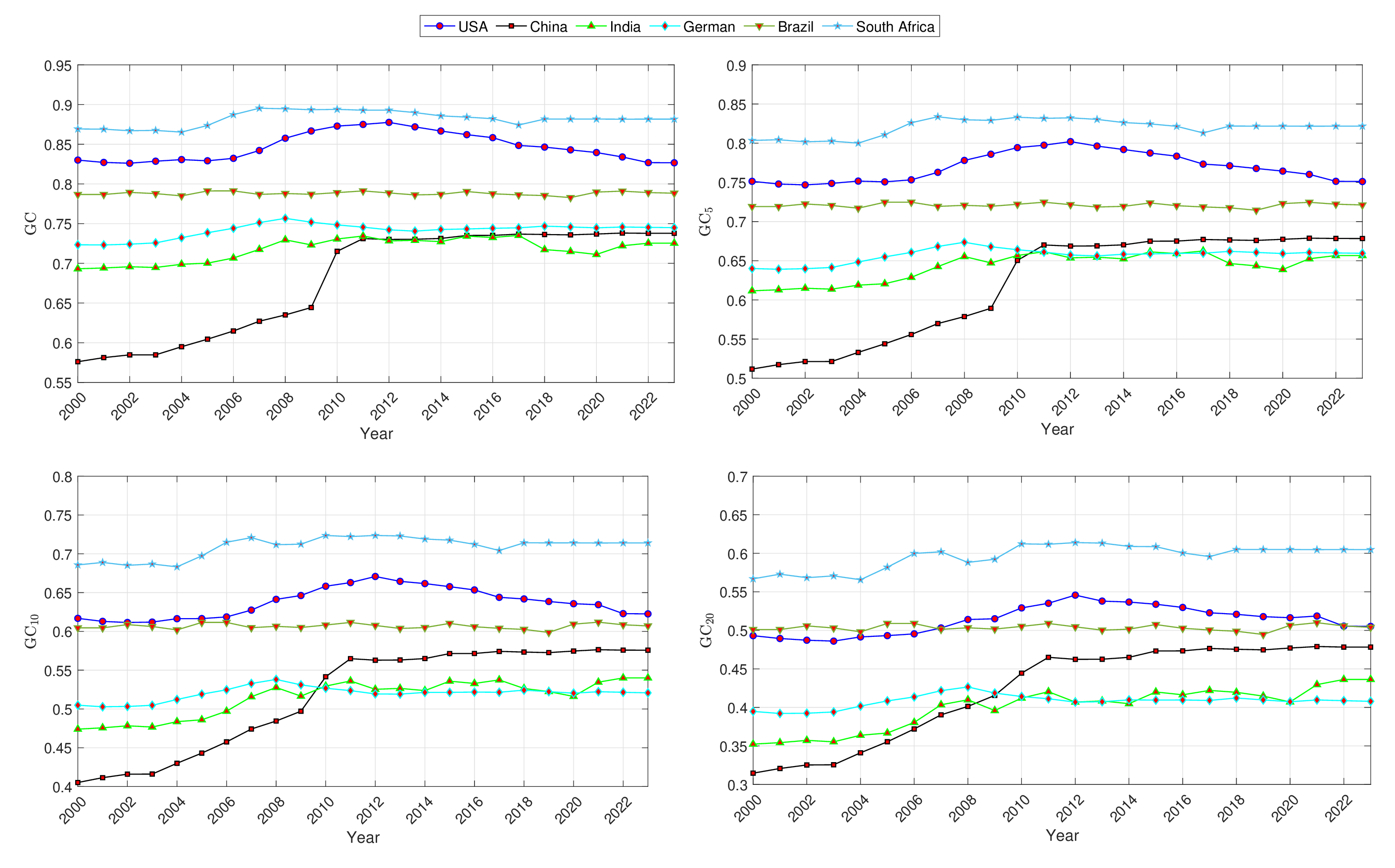}
 \captionsetup{font=small}
 \caption{\small Values of $\mathrm{GC}_n$ based on household wealth for the United States, China, India, German, Brazil and the South Africa from 2000 to 2023, for $n = 2$, 5, 10, and 20.
}\label{fig:GCn}

\end{figure}

\begin{figure}[t!]
\centering
 \includegraphics[width=16cm]{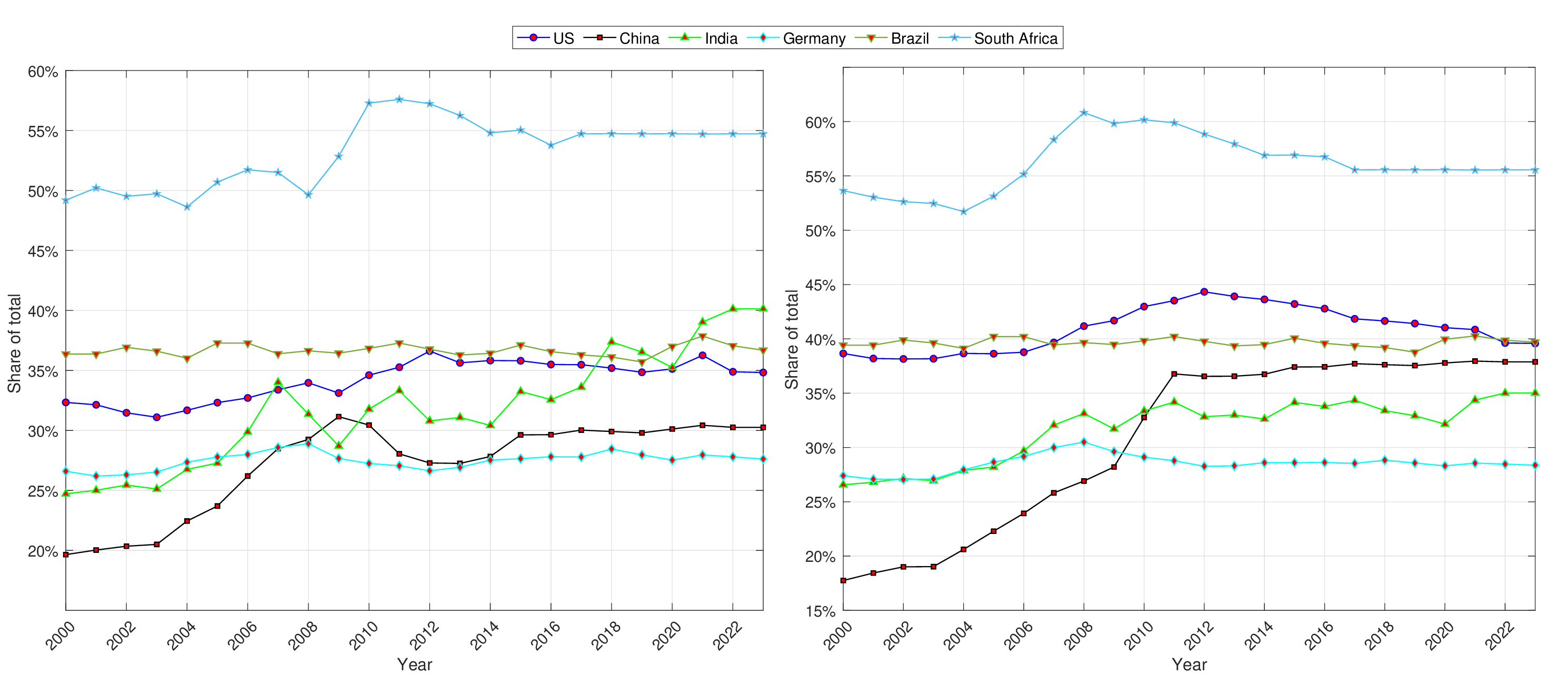}
 \captionsetup{font=small}
 \caption{\small Wealth shares held by the top 1\% (left) and top 10\% (right) in  the  United States, China, India, German, Brazil and the South Africa  from 2000 to 2023.
 }\label{fig:topshare}
\end{figure}

We begin by examining the values of $\mathrm{GC}_n$ for six major economies---the US, China (mainland), India, German, Brazil and South Africa---over the period 2000 to 2023.  The plots for Canada and the UK have already been presented in the Introduction and are therefore omitted here for clarity.  Figure~\ref{fig:GCn} reveals several notable patterns, similar to those observed in Firgure \ref{fig:GC10_country}. For $n = 2$, the values of  $\mathrm{GC}_2$ for China, India, and German are quite close after 2010, indicating broadly similar levels of overall wealth inequality. In contrast,  the South Africa and US consistently displays a much higher $\mathrm{GC}_2$, reflecting its well-documented concentration of wealth. The value for Brazil lies between these two groups, remaining below that of the  US.  As $n$ increases to 5, 10, and 20,   the values of Brazil and the US  become  increasingly similar.
Additionally, China’s $\mathrm{GC}_n$ values rise sharply, surpassing those of India  and  German  and gradually approaching those of the United States. While the trajectories for India and Germany remain closely aligned across all values of $n$, a change emerges after 2010: the valuse of  $\mathrm{GC}_{10}$ and $\mathrm{GC}_{20}$ for  India slightly exceed those of Germany, reversing the pattern seen for $\mathrm{GC}_2$ and $\mathrm{GC}_5$.

To better interpret these results, we plot in Figure~\ref{fig:topshare}  the wealth shares held by the top 1\% and top 10\% in each country. South Africa maintains both high top-share levels and persistently elevated $\mathrm{GC}_n$ values, reaffirming its status as a highly unequal economy.  China's top 10\% share  exceeds those of India and Germany, which is consistent with the sharp rise  in line with the pronounced increase in its $\mathrm{GC}_n$ values as $n$ increases.  We also observe that the top 1\% and top 10\% wealth shares in Brazil are quite close to those in the US, helping to explain the behavior of its $\mathrm{GC}_n$: as $n$ increases, Brazil's index rises significantly and the difference between Brazil and US almost vanishes.
Furthermore, India's top 1\% and top 10\% shares surpass those of Germany after 2004, which explains the shift in their relative $\mathrm{GC}_n$ patterns---while Germany initially has higher values at smaller $n$, India overtakes it as $n$ grows.

\begin{figure}[t!]
\centering
 \includegraphics[width=16cm]{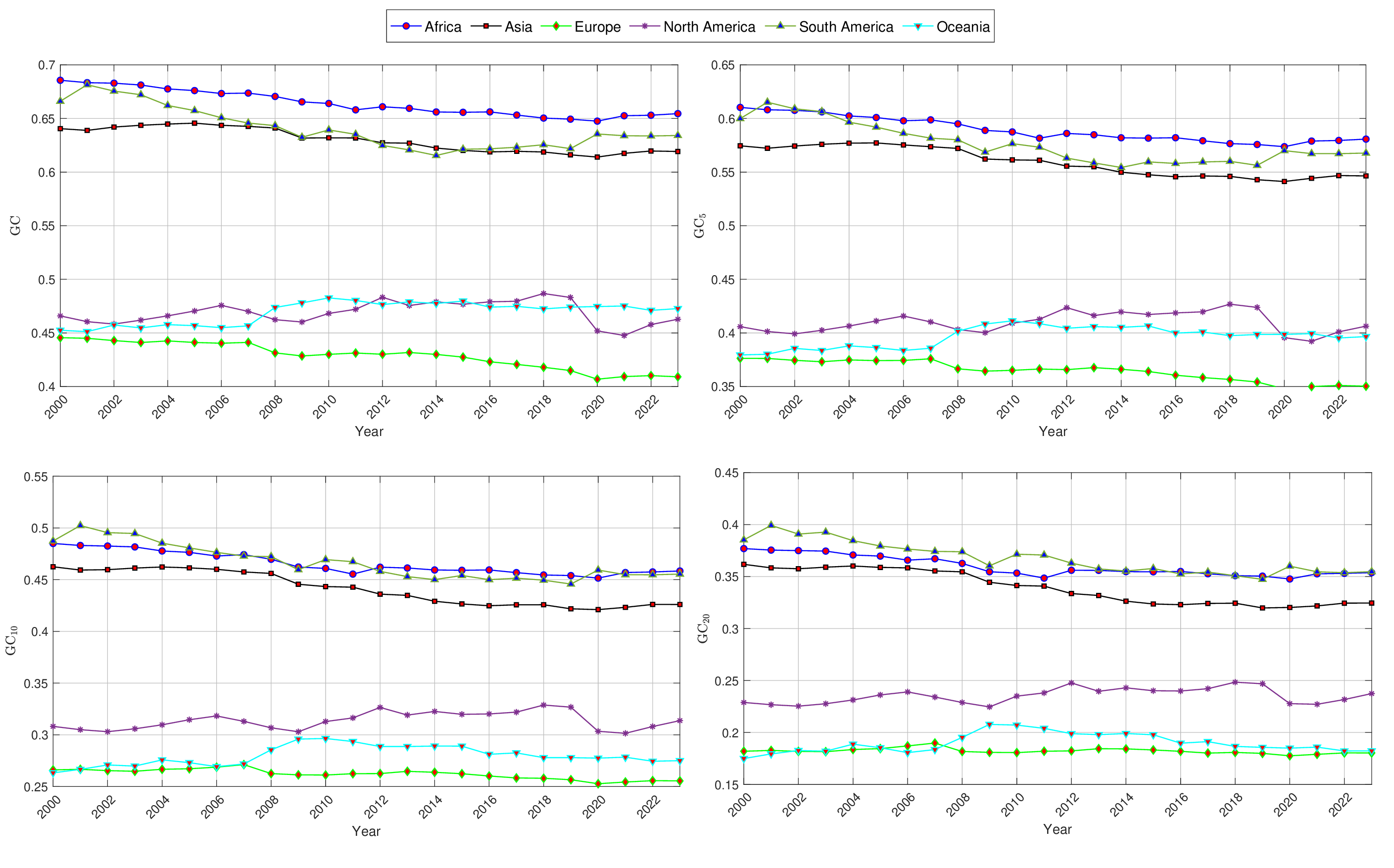}
 \captionsetup{font=small}
 \caption{\small Values of $\mathrm{GC}_n$ based on post-tax national income for Africa, Asia, Europe, North America, South America, and Oceania from 2000 to 2023, for $n = 2$, 5, 10, and 20.
}\label{fig:GCn_region}
\end{figure}

\begin{figure}[t!]
\centering
 \includegraphics[width=16cm]{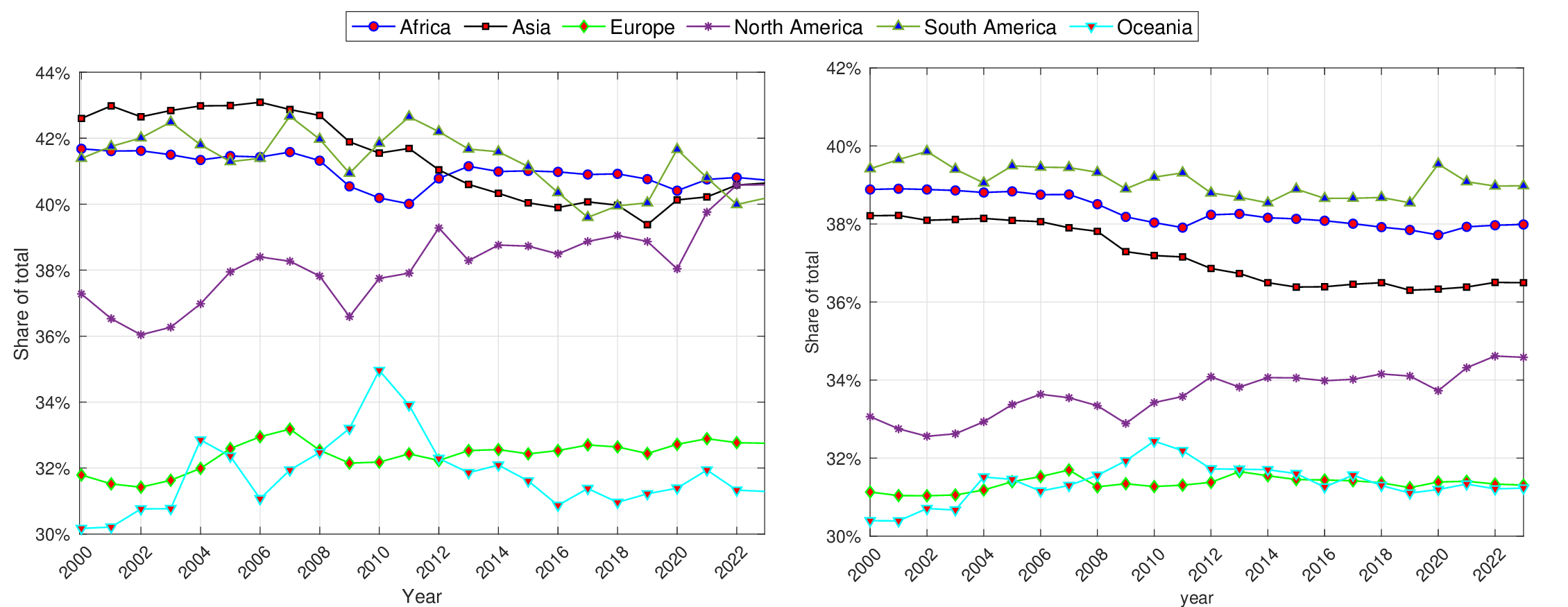}
 \captionsetup{font=small}
 \caption{\small Income shares held by the top 1\% (left) and top 10\% (right) in Africa,  Asia, Europe, North America, South America and Oceania  from 2000 to 2023.}\label{top_share_region}
\end{figure}

Traditionally, GC is  used for measuring income inequality. 
We present  the values of $\GC_n$  for post-tax national income across six major world continents---Africa, Asia, Europe, North America, South America, and Oceania---over the period 2000 to 2023, with $n\in\{2,5,10,20\}$.\footnote{WID provides the full income distribution for each continent, but not for each country.} By Figure \ref{fig:GCn_region}, we can see that  for all $n \geq 2$, $\GC_n$  for Africa, Asia, and South America consistently exceed those of North America, Europe, and Oceania. While $\GC_2$ shows only minor differences among North America, Europe, and Oceania,  the gap  between North America  and Oceania widens progressively with   higher $n$, reflecting greater top-end income concentration consistent with its relatively large shares held by the top 1\% and 10\% income groups as shown in Figure \ref{top_share_region}.  Further,  from the perspective of  $\GC_2$, Asia’s inequality level is close to South America’s and lower than Africa’s. However, when $n \geq 5$, Asia’s $\GC_n$ falls below those of both Africa and South America. South America’s $\GC_n$ values at higher $n$ surpass those of the other two regions, marking it as the region with the highest concentration of top incomes. This aligns with evidence showing that South America’s top 10\% income shares exceed those of both Asia and Africa in Figure \ref{top_share_region}.

These findings highlight the key advantage of $\GC_n$: it uncovers differences that the traditional GC cannot detect. When overall inequality appears similar, higher-order $\GC_n$ values reveal significant disparities in top-tail income concentration. Thus, $\GC_n$ provides a more nuanced tool for assessing income concentration and its evolution over time, offering particular value to policymakers monitoring the impacts of tax reforms or addressing the accumulation of wealth among ultra-high-net-worth individuals.

From the empirical results, if one needs to pick $\GC_n$ for a fixed value of $n$ in complement to the standard $\GC_2$, then  we generally recommend using $\GC_{n}$ with $n=10$, as using a small value of $n$ does not provide much different information than $\GC_2$, and using $n>10$ leads to qualitatively similar results. Nevertheless, in practice it may be helpful to compute $\GC_n$ for several different $n$, or an average of them, to get a more comprehensive picture on the inequality to be measured.

\section{Conclusion}\label{sec:9}

We provided an axiomatic study for a generalization of the classical Gini deviation and Gini coefficient
 as tools for inequality and dispersion measurement. 
We considered an axiom involving expectations of functions of $n \geq 2$ independent and identically distributed copies. By imposing the properties of symmetry, comonotonic additivity, and continuity, we obtained an axiomatic characterization showing that any such measure can be represented as an affine combination of the higher-order Gini deviations.
By transitioning from pairwise to groupwise comparison structures, the generalized indices capture aspects of tail concentration and extremal behavior, different from the traditional two-observation formulation. 

We  analyze the economic properties of the higher-order Gini coefficient, and show that higher-order Gini deviations admit Choquet integral representations, satisfying key properties of coherent deviation and risk measures.  The relationship between the 
high order Gini deviation and the standard deviation is explored. Furthermore, we establish their $n$-observation elicitability,  bridging the gap between descriptive indices and decision-theoretic evaluation.
Empirical analyses using global income and wealth data illustrate the strength of this framework: while the classical GC often masks nuanced disparities across regions or subpopulations, higher-order indices reveal meaningful differences—particularly among top earners. These findings highlight the value of $\mathrm{GC}_n$ as a diagnostic tool for policymakers, regulators, and data scientists facing increasingly complex distributional challenges.

As data environments grow more intricate, the flexibility  of higher-order Gini indices make them promising instruments for capturing refined patterns of inequality, concentration, and systemic risk. In risk management, these indices provide coherent, tail-sensitive alternatives to variance-based risk measures. Their emphasis on upper-tail risks aligns naturally with the objectives of solvency analysis, reinsurance pricing, and regulatory capital assessment. Meanwhile, in machine learning, the multi-observation structure supports the design of loss functions that better capture group-level dispersion, promoting robustness and fairness in algorithmic decision-making—especially in settings involving skewed or heavy-tailed data.


\appendix

\section{Proofs}\label{sec:proofs}


\begin{proof}[Proof of Theorem \ref{th:chara}]
We first prove the equivalence statement. Some statements rely on properties of $\GD_n$ in Proposition \ref{th:basic} in Section \ref{sec:property}, which will be proved later separately. 

\underline{Sufficiency.} 
By the definition of $i$-th order Gini deviation in \eqref{eq:0c}, we have
\begin{align*}
\rho(X)=\sum_{i=1}^n a_i {\rm GD}_i(X)
=\E\left[\sum_{i=2}^n\frac{a_i}{i}\left(\max\{X_1,\dots,X_i\}-\min\{X_1,\dots,X_i\}\right)\right] + \E\left[\frac{a_1}{2}|X_1-X_2|\right].
\end{align*}
This implies that $\rho$ satisfies [A1].
It follows from Proposition \ref{th:basic} that $\rho$ defined by \eqref{eq-affineGD} is a signed Choquet integral, and thus, it follows from Theorem 1 of \cite{WWW20a} that $\rho$ satisfies symmetry, comonotonic additivity, and uniform norm continuity. This completes the proof of sufficiency.

\underline{Necessity.} According to Theorem 1 in \cite{WWW20a}, any mapping $\rho$ satisfying law invariance (implied by [A1]), [A3] and [A4] can be characterized as a signed Choquet integral with an associated distortion function denoted by $h$. We claim  that $h$ is a polynomial of degree at most $n$. To see this, by [A1], there exists a measurable function  $f:\R^n\to\R$ such that  
\begin{align*}
\rho(X)=\E[f(X_1,\dots,X_n)]=\int_0^{\infty} h(\p(X>x))\d x+\int_{-\infty}^0 (h(\p(X>x))-h(1))\d x,~~X\in \X,
\end{align*}
where $X_1,\dots,X_n$ are iid copies of $X$. Let $X$ be a Bernoulli random variable with $\p(X=0)=1-\p(X=1)=1-t$. Define
\begin{align*}
A_i=\{\mathbf z\in \{0,1\}^n: \Vert  \mathbf z\Vert_1 =i\},~~~~i=0, 1,\dots,n,
\end{align*}
where $\Vert  \mathbf z\Vert_1$ represents the $1$-norm of $\mathbf z$.
By straightforward calculation, we have 
\begin{align*}
&h(t)= \rho(X)=\E[f(X_1,\dots,X_n)]
=\sum_{i=0}^n \left(t^{i}(1-t)^{n-i}\sum_{\mathbf z\in A_i}f(\mathbf z)\right),
\end{align*}
which means that $h$ is a polynomial of degree at most $n$. Let $(c_0,c_1,\dots,c_n)\in \R^{n+1}$ satisfy
\begin{align}\label{eq-polynomialdistortion}
h(t)=\sum_{i=0}^n c_i t^i,~~t\in[0,1].
\end{align}
The symmetry [A2] implies that (see \citet[Proposition 2]{WWW20a})
\begin{align*}
h(t)=h(1-t)-h(1),~~t\in[0,1].
\end{align*}
This yields $h(0)= h(1)=0$, and thus $c_0=0$. Moreover,  
\begin{align*}
h(t)
&=\frac{h(t)+h(1-t)}{2}\\
&=\frac12\sum_{i=1}^n c_i(t^i+(1-t)^i)
=\sum_{i=1}^n \left(\frac{-i c_i}{2}\right)\left(\frac{1}{i}\right)(1-t^i-(1-t)^i).
\end{align*} 
By the representation result of ${\rm{GD}}_i$ in terms of signed Choquet integral in Proposition \ref{th:basic} of Section \ref{sec:GDrepresentaton},
$\rho$ admits the form of \eqref{eq-affineGD} with $a_1+a_2=-c_2$ and $a_i=-ic_i/2$ for $i=3,\dots,n$.

\underline{Sufficiency for other axioms.}   
Nonnegativity [A5] follows from \eqref{eq:0c} and $a_1,\dots,a_n\ge 0$.
  [A6]--[A7]
follow directly from properties of signed Choquet integrals in \cite{WWW20a}. 
  [A8]--[A11] follow  the fact that each $\GD_i$ has a concave distortion function in Proposition \ref{th:basic} and $a_1,\dots,a_n\ge 0$; see Theorem 3 of \cite{WWW20a}.  
  To show 
[A12], note that $\E[\max\{X_1,\dots,X_n\}]< \E[X_1+\dots+X_n]= n\E[X_1]$,
and hence $\GD_i(X)/\E[X]\le 1$ for all $X\in \X_+$.
Example \ref{ex:1} below illustrates that  the functionals $\GD_i(X)/\E[X]$ share the same sequence $(X_\epsilon)$ of random variables such that $\GD_i(X_\epsilon)/\E[X_\epsilon]\to 1$ as $\epsilon \downarrow 0$. Therefore,  as their convex combination, $\rho$ satisfies 
$\rho(X_\epsilon)/\E[X_\epsilon]\to 1$, and hence [A12] holds. 
\end{proof}

    
\begin{proof}[Proof of Theorem \ref{prop:GD-p}.]
The properties follow from the corresponding properties of $\GD_n$ in Theorem \ref{th:chara} and the fact that the ratio of a convex function to a positive affine function is quasi-convex (see e.g., \citet[Section 3.4.4]{BV04}).
\end{proof}

\begin{proof}[Proof of Proposition \ref{th:basic}.]
(i) Since $$ \mathbb{P}\left(\max \left\{X_1, \ldots, X_n\right\}  \leq x\right) =(\mathbb{P}(X \leq x))^n=(F_X(x))^n,$$ and  $$ 
\mathbb{P}\left(\min \left\{X_1, \ldots, X_n\right\} \leq x\right)
=1-(\mathbb{P}(X>  x))^n=1-(1-F_X(x))^n,
 $$ we have
\begin{equation}\label{eq:0e}\begin{aligned} \GD_n(X)& = \frac 1n\int_\R  x \d ( F_X(x)^n- 1+ (1-F_X(x))^n)\\& = \int_0^1 F^{-1}_X(t) ( t^{n-1} -(1-t)^{n-1} )\d t.
\end{aligned}\end{equation}
(ii) It follows from \eqref{eq:0e} that  $\GD_n$ is a signed Choquet integral with a concave distortion function
in \eqref{eq:h}. 
\end{proof}

\begin{proof}[Proof of Proposition \ref{prop:comSD}.]
For $\mu\in \R$ and $ \sigma>0$, consider the set $$\mathcal{M}(\mu, \sigma)=\{X \in L^2: \mathbb{E}[X]=\mu, ~ \mathrm{SD}(X)=\sigma\}.$$  Using the H\"older inequality, we obtain, with $h_n$ defined in \eqref{eq:h}, 
$$
\begin{aligned}
\sup _{X \in \mathcal{M}(\mu,\sigma)}\GD_n(X)&= \sup _{X \in \mathcal{M}(0,\sigma)}\GD_n(X)\\&= \sup _{X \in \mathcal{M}(0,\sigma)}  \int_0^1 h_n^{\prime}(t) F^{-1}_X(1-t) \mathrm{d} t \\
& \leqslant \sup _{X\in \mathcal{M}( 0,\sigma)}\left\|h_n^{\prime}\right\|_2\left(\int_0^1\left|F^{-1}_X(1-t)\right|^2 \mathrm{~d} t\right)^{1 / 2}=\sigma\left\|h_n^{\prime}\right\|_2.
\end{aligned} 
$$
Equality is attained if  and only if  $|F^{-1}_X(1-t)|^2$ is  a multiple of $(h'_n)^{1/2}$, leading to $F^{-1}_X(1-t)=h'(t) / \left\|h_n^{\prime}\right\|_2$.
Finally, computing the norm gives
$$\left\|h_n^{\prime}\right\|_2=\sqrt{\frac{2}{2n-1}-\frac{2((n-1)!)^2}{(2n-1)!}},$$ which shows that the upper bound holds and it is sharp.
To show that the lower bound is sharp, consider $X_\epsilon $ in Example \ref{ex:1}, 
for which we have $\GD_n(X_\epsilon)\le \epsilon$, 
and $\mathrm{SD}_n(X_\epsilon)\approx \epsilon^{1/2}$ for $\epsilon>0$ small. 
\end{proof}

\begin{proof}[Proof of Proposition \ref{prop:3}.]
By Proposition \ref{pro:1} in Appendix \ref{sec:monotone}, since $n\leq m$,  we   know $\GD_n(X)/\mathrm{GD}_m(X)\leq 1$ for $X\in L^1_*$.   For the lower bound,    \eqref{eq:r} gives, with $h_n$ defined in \eqref{eq:h},   $$\frac{\GD_n(X)}{\mathrm{GD}_m(X)}\geq \inf_{t\in(0,1)} \frac{h_n(t)}{h_m(t)}$$
   By direct computation, we find that this infimum is attained at $t=0.5$,  yielding
    $$\inf_{t\in(0,1)} \frac{h_n(t)}{h_m(t)}=\frac{h_n(0.5)}{h_m(0.5)}=\frac{m}{n}\frac{1-2^{1-n}}{1-2^{1-m}}.$$
To show that the bound is sharp, consider the two-point distribution
  $$F=(1-p)\delta_x+p\delta_y, ~~~ \text{where}~  x<y ~\text{and}~ p\in(0,1),$$ 
  where $\delta_x$ is the point-mass at $x$. We can compute $$\frac{\GD_n(X)}{\mathrm{GD}_m(X)}=\frac{h_n(p)(y-x)}{h_m(p)(y-x)} =\frac{m(1-p^n-(1-p)^n)}{n(1-p^m-(1-p)^m)}.$$ Taking the limit as $p\downarrow0$, we have $$\lim_{p\downarrow0}\frac{m(1-p^n-(1-p)^n)}{n(1-p^m-(1-p)^m)}= \lim_{p\downarrow0}\frac{-p^{n-1}+(1-p)^{n-1}}{-p^{m-1}+(1-p)^{m-1}}=1.$$
  Hence, the upper bound in \eqref{eq:B1} holds. 
 The lower bound is attained by setting $p=0.5$.
\end{proof}



\begin{proof}[Proof of Theorem \ref{thm:2}.]
It is straightforward to check
$$\begin{aligned}
&\argmin_{x\in \R} \int_{\mathbb{R}^n}S\left(x, y_1, \dots, y_n\right) \mathrm{d} F(y_1) \cdots\mathrm{d} F(y_n)\\&=\argmin_{x\in \R} \mathbb{E}\left[\left(n x-\left(\max\{X_1,\dots,X_n\}-\min\{X_1,\dots,X_n\}\right)\right)^2\right]\\&=\frac{1}{n}  \E \left[\max\{X_1,\dots,X_n\}-\min\{X_1,\dots,X_n\}\right].\end{aligned}
$$
Therefore, $\mathrm{GD}_n$ is $n$-observation $\mathcal{M}^2$-elicitable.
Using the fact that $S$ and $S^*$ differs by a constant term in $x$,  we can see that $S^*$ also elicits $\GD_n$, and it is finite on $\M^1$. 
\end{proof}

\begin{proof}[Proof of Theorem \ref{th-n-elicitability-distortion}.]
Any polynomial $h$ of degree $n$ or less can be written as 
$$ h(t) = a_n (1-t)^n + a_{n-1} (1-t)^{n-1} + \dots + a_1 (1-t) + a_0,~~~t \in [0,1]  , $$
for some  $a_0,a_1,\dots,a_n \in \mathbb{R}$ with $a_0 = -(a_1 + a_2 + \dots + a_n)$. For iid  random variables $X,X_1,\dots,X_n$ with distribtion $F\in \mathcal M^2$,  we have 
\begin{align*} \rho(X)=&\int_0^1 F_X^{-1}(1-t)\d h(t) \\=& \int_0^1 F_X^{-1}(1-t)\d (a_n (1-t)^n+\dots+a_1(1-t)+a_0)   =  -
\sum_{i=1}^n a_i \E[\max\{X_1,\dots,X_i\}] .\end{align*}
For the function $S$ defined by $$S\left(x, y_1, \dots, y_n\right)=\left( x+\sum_{i=1}^n a_i  \max\{y_1,\dots,y_i\}   \right)^2,$$
it is easy to check
\begin{align*}
  &  \argmin_{{x \in \mathbb{R}}} \int_{\mathbb{R}^n}S\left(x, \mathbf y \right) \mathrm{d} F^n(\mathbf y)  =\argmin_{{x \in \mathbb{R}}}\mathbb{E}\left[\left( x+\sum_{i=1}^n a_i  \max\{X_1,\dots,X_i\} \right)^2\right] =\rho(X), 
\end{align*}  
and hence $\rho$ is $n$-observation $\mathcal{M}^2$-elicitable. Removing the square terms from $S$ which does not involve $x$, we get a function $S^*$ that elicits $\rho$ on $\M^1$. 
\end{proof}

\begin{proof}[Proof of Proposition \ref{prop:elicitability-GC}.]
It is straightforward to check
$$\begin{aligned}
&\argmin_{x\in \R} \int_{\mathbb{R}^n}S\left(x, y_1, \dots, y_n\right) \mathrm{d} F(y_1) \cdots\mathrm{d} F(y_n)\\&=\argmin_{x\in \R} \mathbb{E}\left[x^2 X_1-\frac{2x}{n}\max\{X_1,\dots,X_n\}-\min\{X_1,\dots,X_n\}\right]\\&=\frac{1}{n}  \E \left[\max\{X_1,\dots,X_n\}-\min\{X_1,\dots,X_n\}\right]/\E[X_1]= \mathrm{GC}_n(X).\end{aligned}
$$
Therefore, $\mathrm{GC}_n$ is $n$-observation $\mathcal{M}^1$-elicitable.
\end{proof}

\begin{proof}[Proof of Theorem \ref{thm:5}]
The Law of Large Numbers yields $\widehat x_N\stackrel{\mathbb{P}}{\rightarrow} \E[X]$. 
 By Theorem 2.6 of \cite{KSZ14}, 
 the empirical estimator for a finite convex risk measure on $L^1$ is consistent, that is,
 $\widehat \GD_n (N)  + \widehat x_N \stackrel{\mathbb{P}}{\rightarrow} \GD_n(X) +\E[X]$, and this gives  $\widehat \GD_n (N)  \stackrel{\mathbb{P}}{\rightarrow} \GD_n(X)$. 
Moreover,  by the Continuous Mapping Theorem, since the mapping $(a,b)\to\frac{a}{b}$ is continuous at $(\GD_n(X),\E[X])$ with $\E[X]>0$,  we have $\widehat{\GC}_n(N) \stackrel{\mathbb{P}}{\rightarrow} \GC_n(X)$. 

Next, we will show the asymptotic normality.
Let $B=(B_t)_{t\in [0,1]}$ be a standard Brownian bridge, and let $d_N= \sqrt{N} (\widehat{\GD}_n(N)-\GD_n(X) ) $,
  $e_N =\sqrt{N} (\widehat x_N-\E[X])  $, 
  and $g_N= \sqrt{N} (\widehat{\GC}_n(N)-\GC_n(X) )) $. 
Note that $\GD_n$
can be written as an integral of the quantile, that is, 
$$
\GD_n(X) = \int_0^1 F^{-1}(t) ( t^{n-1} -(1-t)^{n-1} )\d t.
$$
Denote by $A_N$ 
the empirical quantile process, that is,
$$
A_N(t)= \sqrt{N} (\widehat F^{-1}_N(t) - F^{-1}(t)), ~~~t\in (0,1).
$$
It follows that 
$$ d_N = 
\int_0^1 A_N(t)  ( t^{n-1} -(1-t)^{n-1} ) \d t.
$$
It is well known that, under Assumption \ref{assump:1}, as $N\to \infty$, $A_N$
converges to the Gaussian process $B/\tilde f$   in $L^\infty[\delta,1-\delta]$ for any $\delta>0$ (see e.g.,  \cite{DGU05}).
This yields 
 $$
\int_\delta^{1-\delta} A_N(t) ( t^{n-1} -(1-t)^{n-1} ) \d t \stackrel{\mathrm{d}}\rightarrow  
\int_\delta^{1-\delta}  \frac{B_t}{\tilde f(t)}( t^{n-1} -(1-t)^{n-1} )\d t.$$ 
Next, we  verify      \begin{equation} 
\label{eq:joint-conv3}
\int_\delta^{1-\delta} ( t^{n-1} -(1-t)^{n-1} ) \frac{B_t}{\tilde f(t)} \d t\to  \int_0^1 \frac{B_t}{\tilde f(t)}( t^{n-1} -(1-t)^{n-1} )\d t \mbox{~~~as $\delta \downarrow 0$}.
\end{equation}
Denote by $w_t=t(1-t)$. 
Since 
$X\in L^{\gamma}$, we have,  for some $C>0$, 
$$
|F^{-1}(t)| \le C w_t^{-1/\gamma}~~\text{and}~ \frac{1}{\tilde f(t) } =
\frac{\d F^{-1}(t)} {\d t} \le C  w_t ^{-1/\gamma -1}. 
$$ 
Note that $B_t=o_{\p}(w_t^{1/2-\epsilon})$ for any $\epsilon>0$ as $t\to 0$ or $1$. Hence,  
for some $\eta>0$,
$$
\left| \frac{B_t}{\tilde f(t)}( t^{n-1} -(1-t)^{n-1} )\right| =o_{\p} (w_t^{\eta-1}) \mbox{~~~for $t\in (0,1)$},
$$
and this 
guarantees \eqref{eq:joint-conv3}. 
Using the convariance property of the Brownian bridge, that is,
$\mathrm{Cov}(B_t, B_s)=s-s t$ for $s<t$, we have
\begin{align*}&\var\left[\int_0^1 \frac{B_s( s^{n-1} -(1-s)^{n-1} )}{\tilde f(s)} \mathrm{d} s\right]\\ & =\mathbb{E}\left[\int_0^1 \int_0^1 \frac{( s^{n-1} -(1-s)^{n-1} )( t^{n-1} -(1-t)^{n-1} )B_s B_t}{\tilde f(s) \tilde f(t)} \mathrm{d} t \mathrm{d} s\right]\\ & = \int_0^1 \int_0^1 \frac{( s^{n-1} -(1-s)^{n-1} )( t^{n-1} -(1-t)^{n-1} )(s \wedge t-s t)}{\tilde f(s) \tilde f(t)} \mathrm{d} t \mathrm{d} s. \end{align*} 
Thus,  $
\sqrt{N}(\widehat{\GD}_n(N)-\GD_n(X)) \stackrel{\mathrm{d}}{\rightarrow} \mathrm{N}(0, \sigma^2_{\GD_n})
$ in which $\sigma^2_{\GD_n}$ is given by \eqref{eq:sigma_GD}.

Note that $$g_N=\sqrt{N}\left(\frac{\widehat{\GD}_n(N)}{\widehat x_N}-\frac{\GD_n(X)}{\E[X]}\right).$$
Since $$ d_N = 
\int_0^1 A_N(t)  ( t^{n-1} -(1-t)^{n-1} ) \d t, ~~~ e_N = 
\int_0^1 A_N(t)  \d t,
$$
by Delta method, we have 
$$
g_N = \textbf{R}^\top \textbf{L}+o_p(1),
$$
where $$ \textbf{R}=\left(\frac{1}{\E[X]}, -\frac{\GD_n(X)}{(\E[X])^2}\right), ~~~~\textbf{L}= (d_N, e_N).$$ 
Since $$
(d_N,e_N)   \stackrel{\mathrm{d}}{\rightarrow}\left(\int_0^1 \frac{B_s }{\tilde f(s)} ( s^{n-1} -(1-s)^{n-1} ) \mathrm{d} s , \int_0^1 \frac{B_s}{\tilde f(s)} \mathrm{d} s\right),$$ 
we have $$g_N \stackrel{\mathrm{d}}{\rightarrow}\frac{1}{\E[X]} \int_0^1 \frac{B_s }{\tilde f(s)} \left( s^{n-1} -(1-s)^{n-1}-\GC_n(X)\right) \mathrm{d} s.$$
Using the convariance property of the Brownian bridge,  we have
\begin{align*}&\var\left[\frac{1}{\E[X]} \int_0^1 \frac{B_s }{\tilde f(s)} \left( s^{n-1} -(1-s)^{n-1}- \GC_n(X)\right) \mathrm{d} s\right]\\ & =\mathbb{E}\left[\int_0^1 \int_0^1 \frac{( s^{n-1} -(1-s)^{n-1} -\GC_n(X))( t^{n-1} -(1-t)^{n-1}- \GC_n(X))B_s B_t}{(\E[X])^2\tilde f(s) \tilde f(t)} \mathrm{d} t \mathrm{d} s\right]\\ & = \int_0^1 \int_0^1 \frac{( s^{n-1} -(1-s)^{n-1}-\GC_n(X) )( t^{n-1} -(1-t)^{n-1}-\GC_n(X) )(s \wedge t-s t)}{(\E[X])^2\tilde f(s) \tilde f(t)} \mathrm{d} t \mathrm{d} s. \end{align*} 
Thus,  $
\sqrt{N}(\widehat{\GC}_n(N)-\GC_n(X)) \stackrel{\mathrm{d}}{\rightarrow} \mathrm{N}(0, \sigma^2_{\GC_n})
$ in which $\sigma^2_{\GC_n}$ is given by \eqref{eq:sigma_GC}.
\end{proof}

\section{Explicit GD and GC formulas   for parametric distributions}\label{sec:7}

In this section, we derive explicit formulas or integral representations for $\GD_n$ and $\GC_n$ for several commonly used parametric distributions, including Bernoulli,  Beta, log-normal, exponential, and Pareto distributions.
\begin{itemize}
\item [(i)] Bernoulli distribution:
Let $X\sim \mathrm B(p)$, where $p\in(0,1)$. Then $$
\GD_n(X) = \int_{1-p}^1  ( t^{n-1} - (1 - t)^{n-1} ) \mathrm{d}t=\frac{1}{n}(1-p^n-(1-p)^n).
$$
Since $\E[X]=p$, giving 
$$
\GC_n(X) =\frac{1-p^n-(1-p)^n}{np}.
$$  Figure \ref{fig:Ber} presents the curves of  $\mathrm{GD}_n$ and $\mathrm{GC}_n$ as 
$p$ and 
$n$ vary. 
\begin{figure}[t!]
\centering
 \includegraphics[width=16cm]{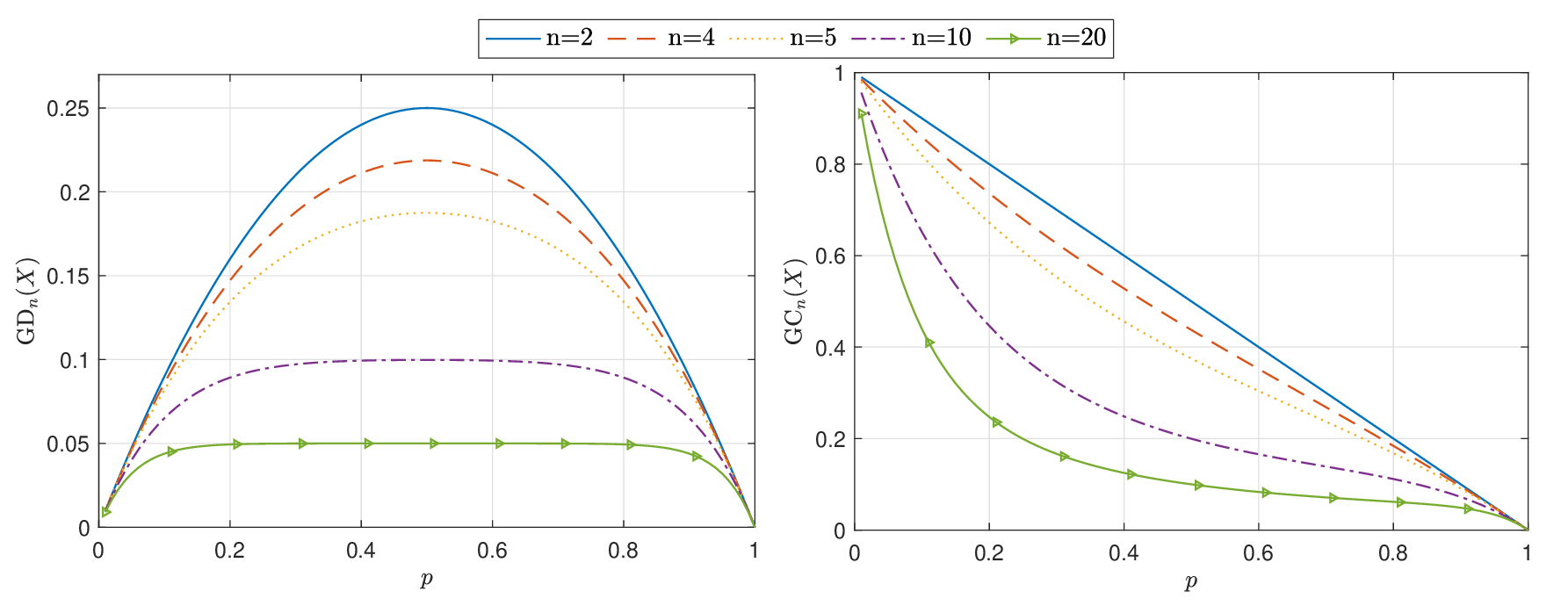}
 \captionsetup{font=small}
 \caption{ \small   Values of  $\mathrm {GD}_n$ and $\mathrm {GC}_n$ for  
$X \sim \mathrm B(p)$  with $p \in(0, 1)$}\label{fig:Ber}
\end{figure}
\item[(ii)] Beta distribution: 
For a Beta-distributed random variable $X \sim \text{Beta}(a, b)$ with shape parameters $a, b > 0$, the CDF is given by the regularized incomplete Beta function:
$$
F_X(x) = I_x(a, b) = \frac{\int_0^x u^{a-1} (1 - u)^{b-1} \mathrm{d}u}{B(a, b)},
$$
where $B(a, b) = \int_0^1 u^{a-1} (1 - u)^{b-1} \mathrm{d}u $ is the Beta function. The quantile function is denoted as:
$
F_X^{-1}(t) = I_t^{-1}(a, b),
$
where $I_t^{-1}(a, b)$  is the inverse regularized incomplete Beta function. Thus, $\GD_n(X)$ for the Beta distribution is:
$$
\GD_n(X) = \int_0^1 I_t^{-1}(a, b) \left( t^{n-1} - (1 - t)^{n-1} \right) \mathrm{d}t.
$$
Since $\mathbb{E}[X] = a/(a + b)$, we obtain
$$
\GC_n(X) = \frac{\GD_n(X)(a+b)}{a}.
$$ Figures \ref{fig:beta1}  and \ref{fig:beta2} present the curves of  $\mathrm{GD}_n$ and $\mathrm{GC}_n$ as 
$a$ 
$b$, and $n$ vary. 

\begin{figure}[t!]
\centering
 \includegraphics[width=16cm]{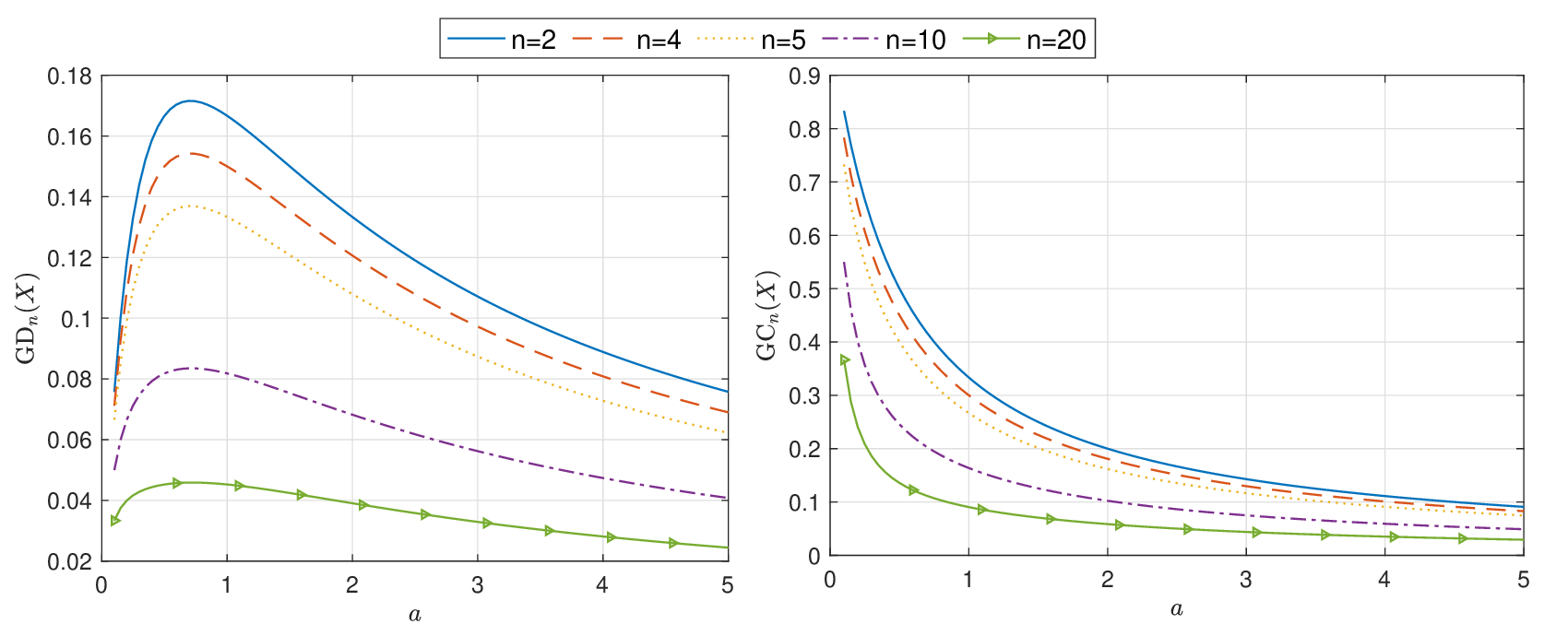}
 \captionsetup{font=small}
 \caption{ \small  Values of  $\mathrm {GD}_n$ and $\mathrm {GC}_n$ for  
$X \sim \text{Beta}(a, 1)$  with $a\in(0, 5)$}\label{fig:beta1}
\end{figure}
\begin{figure}[t!]
\centering
 \includegraphics[width=16cm]{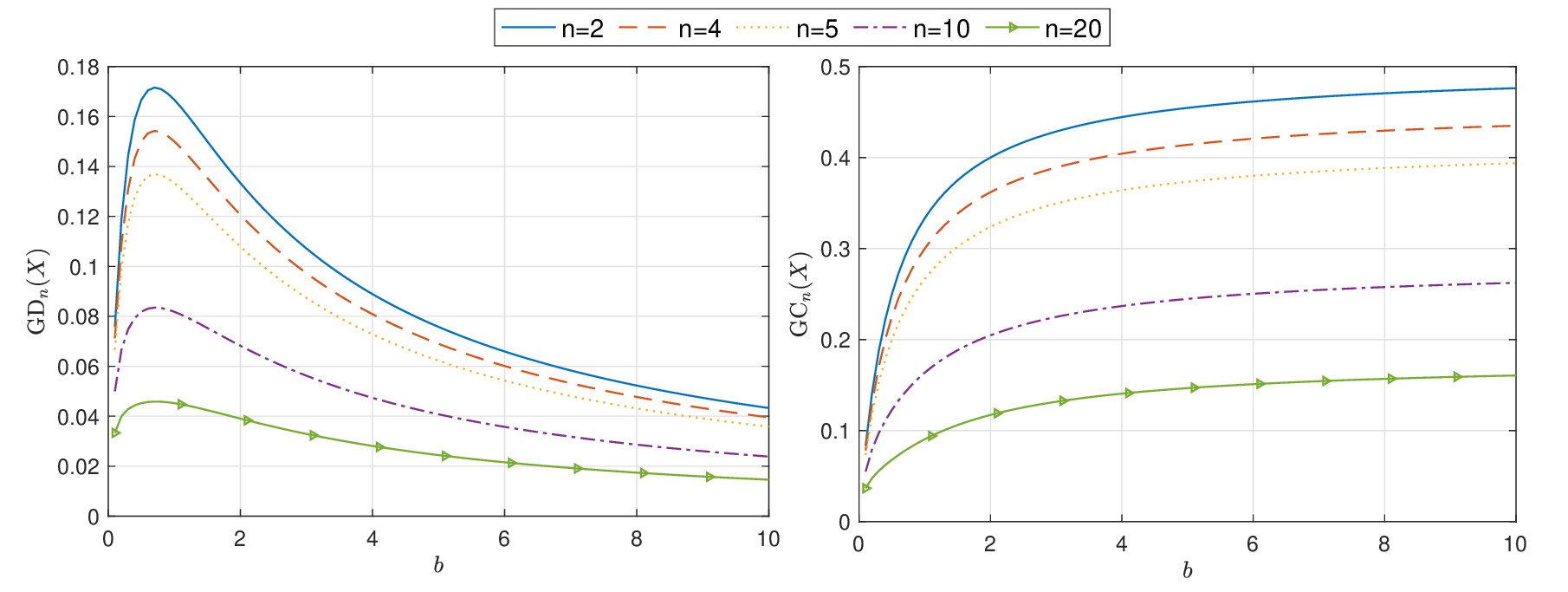}
 \captionsetup{font=small}
 \caption{ \small   Values of  $\mathrm {GD}_n$ and $\mathrm {GC}_n$ for  
$X \sim \text{Beta}(1, b)$  with $b\in(0, 10)$}\label{fig:beta2}
\end{figure}
\item[(iii)]  Log-Normal distribution:
For a log-normal random variable $X \sim \mathrm{LN}(\mu, \sigma^2)$, the CDF is given by
$$
F_X(x) = \Phi\left( \frac{\ln x - \mu}{\sigma} \right),
$$
where $\Phi$ is the standard normal CDF. The quantile function  is given by
$$
F_X^{-1}(t) = \exp(\mu + \sigma \Phi^{-1}(t)).
$$
Thus, $\GD_n(X)$ is
$$
\GD_n(X) =  \int_0^1 \exp(\mu + \sigma \Phi^{-1}(t)) \left( t^{n-1} - (1 - t)^{n-1} \right) \mathrm{d}t.
$$
The expectation of $X$ is $\mathbb{E}[X] = e^{\mu + \sigma^2/2}$, giving
$$
\GC_n(X) = \frac{\GD_n(X)}{e^{\mu + \sigma^2/2}}=\int_0^1 \exp\left( \sigma \Phi^{-1}(t)- \frac{\sigma^2} {2}\right) \left( t^{n-1} - (1 - t)^{n-1} \right) \mathrm{d}t,
$$
which is independent in $\mu$. Figures \ref{fig:LN1}  and \ref{fig:LN2} present the curves of  $\mathrm{GD}_n$ and $\mathrm{GC}_n$ as 
$\mu$ 
$\sigma$, and $n$ vary. 
\begin{figure}[t!]
\centering
 \includegraphics[width=16cm]{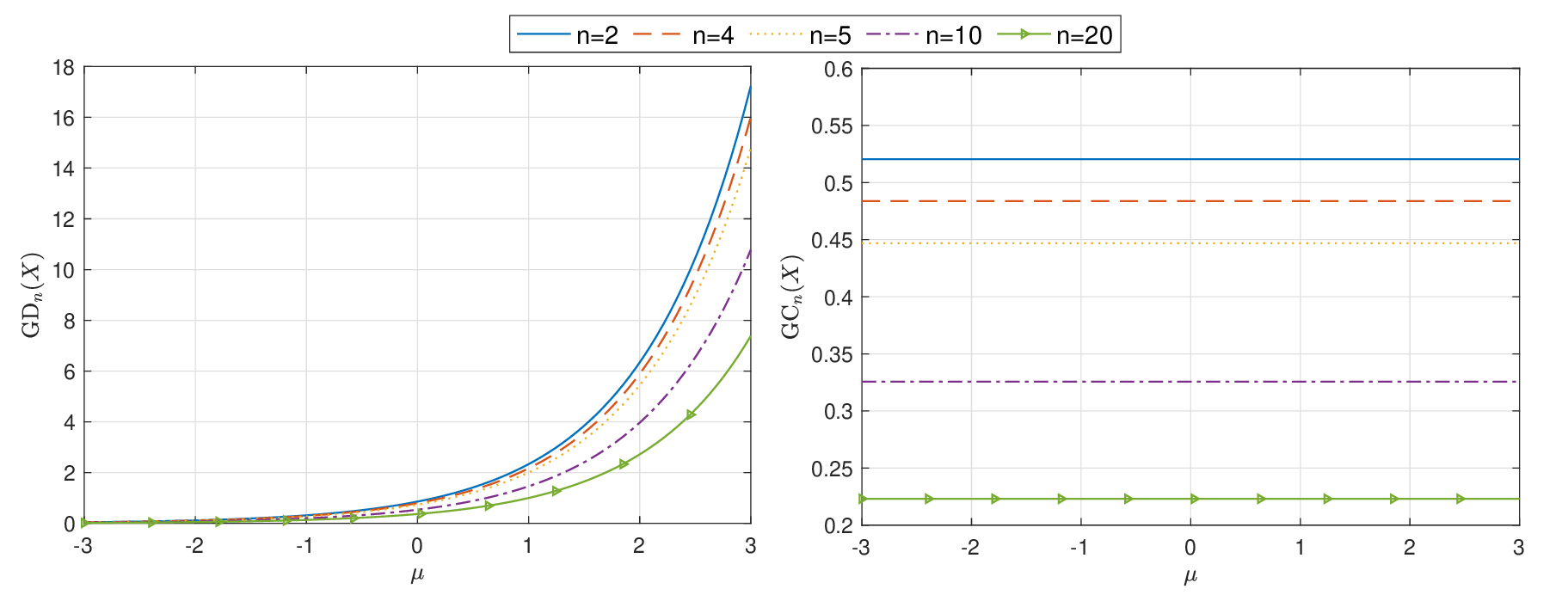}
 \captionsetup{font=small}
 \caption{ \small  Values of  $\mathrm {GD}_n$ and $\mathrm {GC}_n$ for  
$X \sim \mathrm{LN}(\mu, 1)$  with $\mu\in(-3, 3)$}\label{fig:LN1}
\end{figure}
\begin{figure}[t!]
\centering
 \includegraphics[width=16cm]{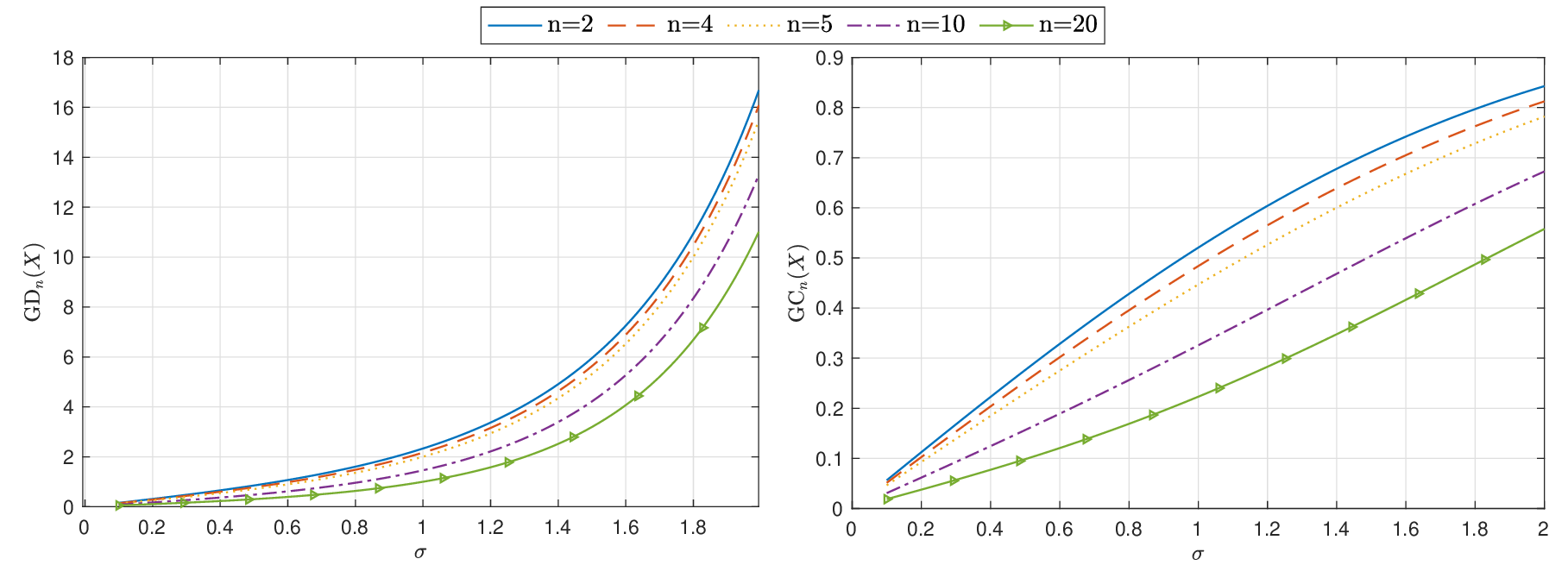}
 \captionsetup{font=small}
 \caption{ \small   Values of  $\mathrm {GD}_n$ and $\mathrm {GC}_n$ for  
$X \sim \mathrm{LN}(1, \sigma^2)$  with $\sigma\in(0, 2)$}\label{fig:LN2}
\end{figure}
\item[(iv)] Exponential distribution:  For an exponential random variable $ X \sim \text{exp}(\lambda)$, the CDF is given by
$
F_X(x) = 1 - e^{-\lambda x}, ~x \geq 0.$
The quantile function  is
$$
F_X^{-1}(t) = -\frac{1}{\lambda} \ln(1 - t), \quad t \in [0, 1).
$$
Thus, we have 
$$\begin{aligned}
\mathrm{GD}_n(X) &= \int_0^1 \left( -\frac{1}{\lambda} \ln(1 - t) \right) \left( t^{n-1} - (1 - t)^{n-1} \right)  \d t\\&=\frac {1} {n\lambda}  \sum_{k=1}^{n-1} \frac {1} {k}.\end{aligned}
$$
The expectation of $X$ is $\mathbb{E}[X] = 1/\lambda$, giving
$$
\GC_n(X) = \frac {1} {n} \sum_{k=1}^{n-1} \frac {1} {k},
$$
which is independent in $\lambda$. 
 Figure \ref{fig:exp} presents the curves of  $\mathrm{GD}_n$ and $\mathrm{GC}_n$ as 
$\lambda$ 
 and $n$ vary.  
\begin{figure}[t!]
\centering
 \includegraphics[width=16cm]{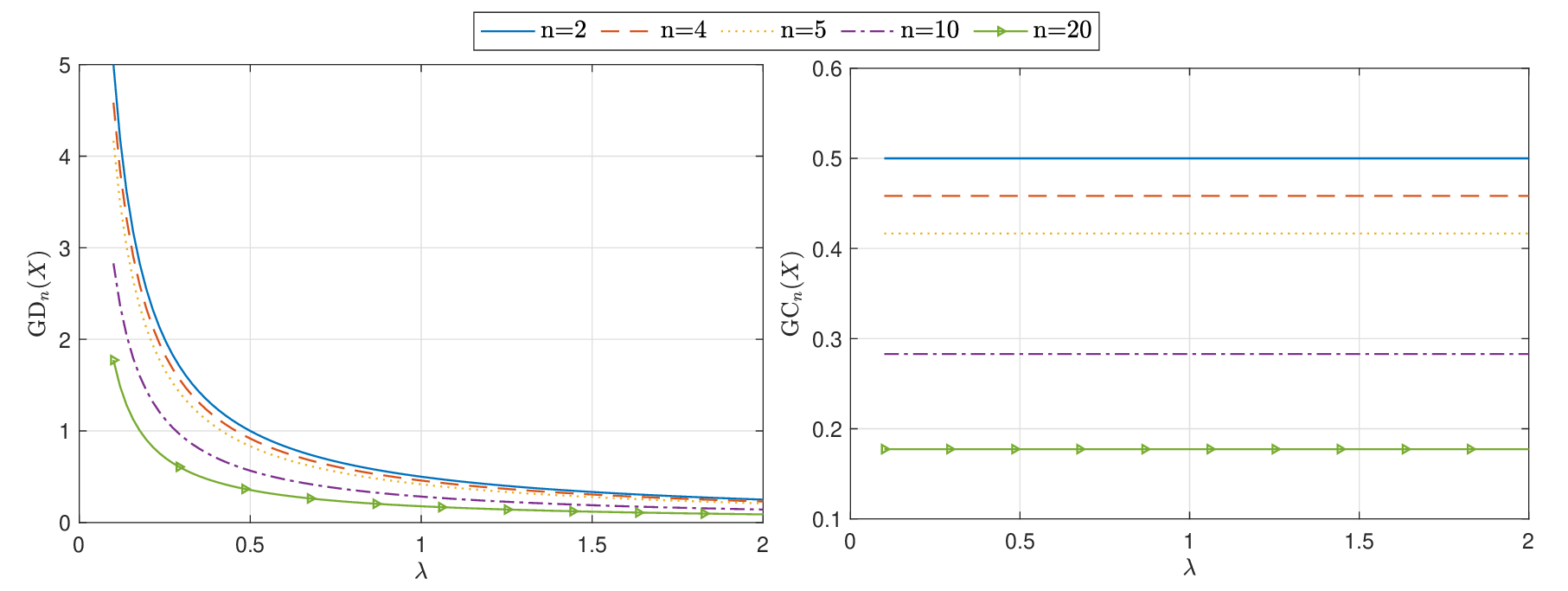}
 \captionsetup{font=small}
 \caption{ \small   Values of  $\mathrm {GD}_n$ and $\mathrm {GC}_n$ for  
$X \sim \text{exp}(\lambda)$  with $\lambda\in(0, 2)$}\label{fig:exp}
\end{figure}
\item[(v)] Pareto distribution:  
For a Pareto-distributed random variable $X \sim \text{Pareto}(\alpha, x_m)$, the  CDF is given by  
$$
F_X(x) = 1 - \left(\frac{x_m}{x}\right)^{\alpha}, \quad x \geq x_m.
$$ 
The quantile function   is  
$$
F_X^{-1}(t) = x_m (1 - t)^{-1/\alpha}, \quad t \in [0, 1).
$$
Thus, we have  
$$
\begin{aligned}
\mathrm{GD}_n(X) &= \int_0^1 \left( x_m (1 - t)^{-1/\alpha} \right) \left( t^{n-1} - (1 - t)^{n-1} \right) \d t \\
&=  x_m\left(B(n, 1-1/\alpha)-\frac{1}{n-1/\alpha}\right).
\end{aligned}
$$
The expectation of $X$ is $\mathbb{E}[X] = \alpha x_m/(\alpha - 1)$ for $\alpha > 1$ (for $ \alpha \leq 1 $, $ \mathbb{E}[Y] $ is infinite), giving 
$$
\GC_n(X) = \frac{\alpha-1}{\alpha}\left(B(n, 1-1/\alpha)-\frac{1}{n-1/\alpha}\right), 
$$
where $B(n,1-1/\alpha)$ is the Beta function.  It is clear that $\GC_n(X)$  is independent in $x_m$.  Figures \ref{fig:par1} and \ref{fig:par2} present the curves of  $\mathrm{GD}_n$ and $\mathrm{GC}_n$ as 
$\alpha$ 
 and $n$ vary.  
\begin{figure}[t!]
\centering
 \includegraphics[width=16cm]{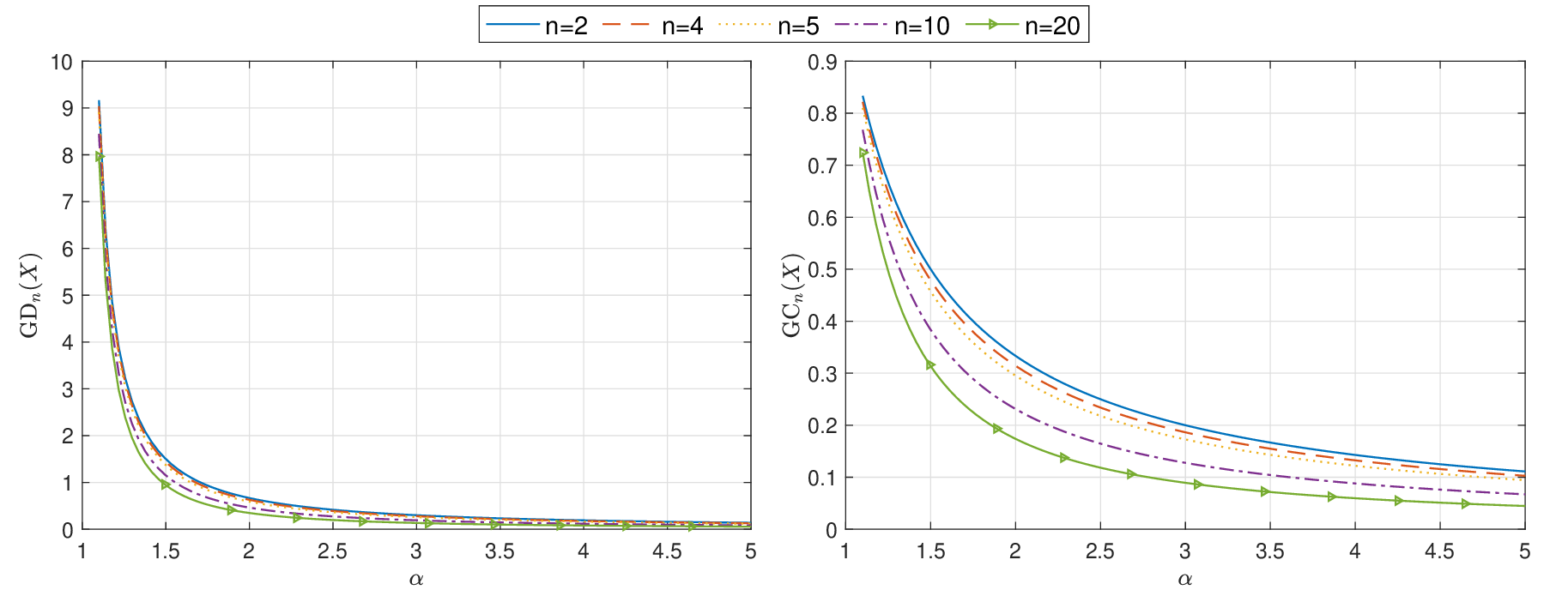}
 \captionsetup{font=small}
 \caption{ \small  Values of  $\mathrm {GD}_n$ and $\mathrm {GC}_n$ for  
$X \sim\text{Pareto}(\alpha, 1)$   with $\alpha\in(1, 5)$}\label{fig:par1}
\end{figure}
\begin{figure}[t!]
\centering
 \includegraphics[width=16cm]{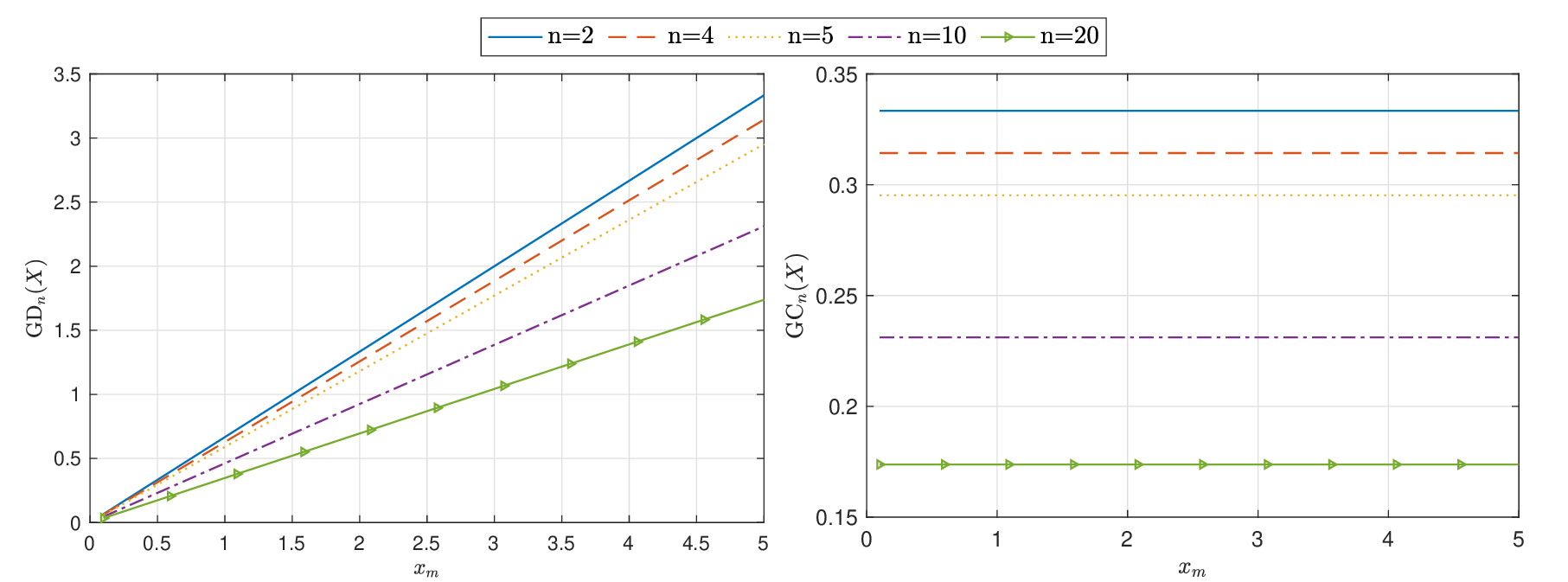}
 \captionsetup{font=small}
 \caption{ \small   Values of  $\mathrm {GD}_n$ and $\mathrm {GC}_n$ for  
$X \sim\text{Pareto}(2, x_m)$  with $x_m\in(0, 5)$}\label{fig:par2}
\end{figure}
\end{itemize}

\section{Monotonicity with respect to the order}
\label{sec:monotone}
In this section, we discuss how $\GD_n$ and $\GC_n$ changes as the parameter $n$ varies. 
Recall that $\GD_3=\GD_2$.  
  The situation for $n \ge 4$ is however different. We get from \eqref{eq:h}, $$h_4(t)=  \frac 12 t (1-t) 
  (2-t+t^2)  \le h_3(t)=h_2(t).$$ Therefore, $\GD_4\le \GD_3=\GD_2$, and $\GD_4\ne \GD_3$.
  On the other hand, since $2-t+t^2\ge 7/4$ for $t\in [0,1]$, we have 
  $
  \GD_4(X)/\GD_3(X)  \in [7/8,1]
  $
  whenever $\GD_3(X) \ne 0$ (see also Proposition \ref{prop:3} below).  
In the following proposition, we can show that $\GD_n$ decreases in $n$, and this result was obtained by \cite{GROG24} with a different proof.
  \begin{proposition}\label{pro:1}
For $X\in L^1$ (resp.~$X\in L^1_+$), 
 $\GD_n(X)$ (resp.~$\GC_n(X)$) decreases   to $0$ as $n\to\infty$. 
  \end{proposition}
  \begin{proof}
  We only prove the statement for $\GD_n$, as the proof for $\GC_n$ is identical. 
  Let   $n \ge2 $ and  $X_1,\dots,X_{n+1}$ be iid copies of $X$. 
  Define $A_k = \max \left\{X_1,  \dots X_k\right\} $  and $B_k = \min \left\{X_1, \dots X_k\right\} $ for $k\in [n+1]$. By symmetry,     for each $i \in [n+1]$,  we have 
$$\begin{aligned} &\mathbb E\left[ A_n-B_n \mid  A_{n+1}, B_{n+1} \right]
 =\mathbb E\left[\max_{j\in [n+1]\setminus \{i\}} X_j -\min_{j\in [n+1]\setminus \{i\}} X_j   \mid A_{n+1}, B_{n+1}\right].
 \end{aligned}$$ By setting $R_n = A_n - B_n$, we have  
$$  
\sum_{i=1}^{n+1} \E\left[\max_{j\in [n+1]\setminus \{i\}} X_j -\min_{j\in [n+1]\setminus \{i\}} X_j   \mid  A_{n+1}, B_{n+1} \right] =\E[(n+1) R_n\mid  A_{n+1}, B_{n+1} ]. 
$$
Note 
that 
$$
\sum_{i=1}^{n+1} \max_{j\in [n+1]\setminus \{i\}} X_j
= n A_{n+1} + X_{n:n+1}  
\mbox{~~~and~~~}\sum_{i=1}^{n+1} \max_{j\in [n+1]\setminus \{i\}} X_j
= n B_{n+1} + X_{2:n+1}, 
$$
where $X_{n:n+1}$ (resp.~$X_{2:n+1}$) is the second largest (resp.~second smallest) value in $X_1,\dots,X_{n+1}$.
Since $X_{n:n+1}\ge X_{2:n+1}$, 
We have $\E\left[(n+1) R_n \mid A_{n+1}, B_{n+1}\right]\geq\E\left[n R_{n+1} \mid A_{n+1}, B_{n+1}\right]=n R_{n+1}$.  It  then follows that  $$ 
 \E\left[\left.\frac{1}{n} R_n \,\right\rvert\,  A_{n+1}, B_{n+1}\right]\geq\frac{1}{n+1} R_{n+1},$$ and thus $$
 \frac{1}{n} \E [R_n]\geq \frac{1}{n+1} \E [R_{n+1}],
$$ which completes the proof of the first statement. Furthermore, by the dominated convergence theorem, we have
$$
\lim_{n \to \infty} \GD_n(X) = \int_0^1 F_X^{-1}(t) \lim_{n \to \infty} \left( t^{n-1} - (1-t)^{n-1} \right) \d t = 0,
$$
showing the second statement. 
\end{proof}
In the following example, we compare the behavior of $ \GC $  and $ \GC_n$  for two groups of different distributions: (i) $ X \sim \mathrm{LN}(\mu,\sigma)$ and   $Y \sim \text{Pareto}(\alpha, x_m)$, and (ii) $ X \sim \text{Beta}(a,b)$ and   $Y \sim \text{Pareto}(\alpha, x_m)$. The specific distribution details and the expression for computing $\GC_n$
  are provided in Appendix \ref{sec:7}. Note that $\GC_n$ is independent of $\mu$ for the log-normal distribution and independent of $x_m$ for the Pareto distribution. 
\begin{figure}[t!]
\centering
 \includegraphics[width=16cm]{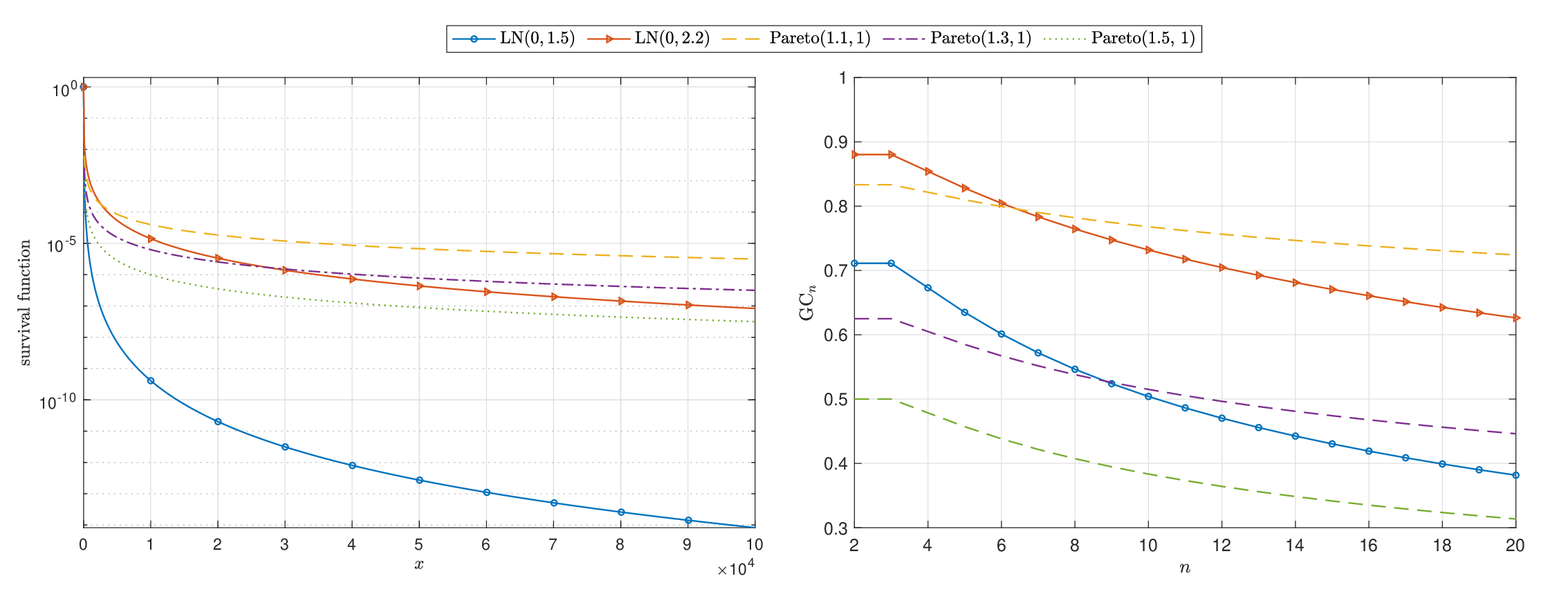}
 \captionsetup{font=small}
 \caption{ \small Survival functions (left panel)  and  values of  $\mathrm {GC}_n$ (right panel) for 
$\mathrm{LN}(0,\sigma)$ and $\mathrm{Pareto}(\alpha,1)$}\label{fig:1}
\end{figure}

\begin{figure}[t!]
\centering
 \includegraphics[width=16cm]{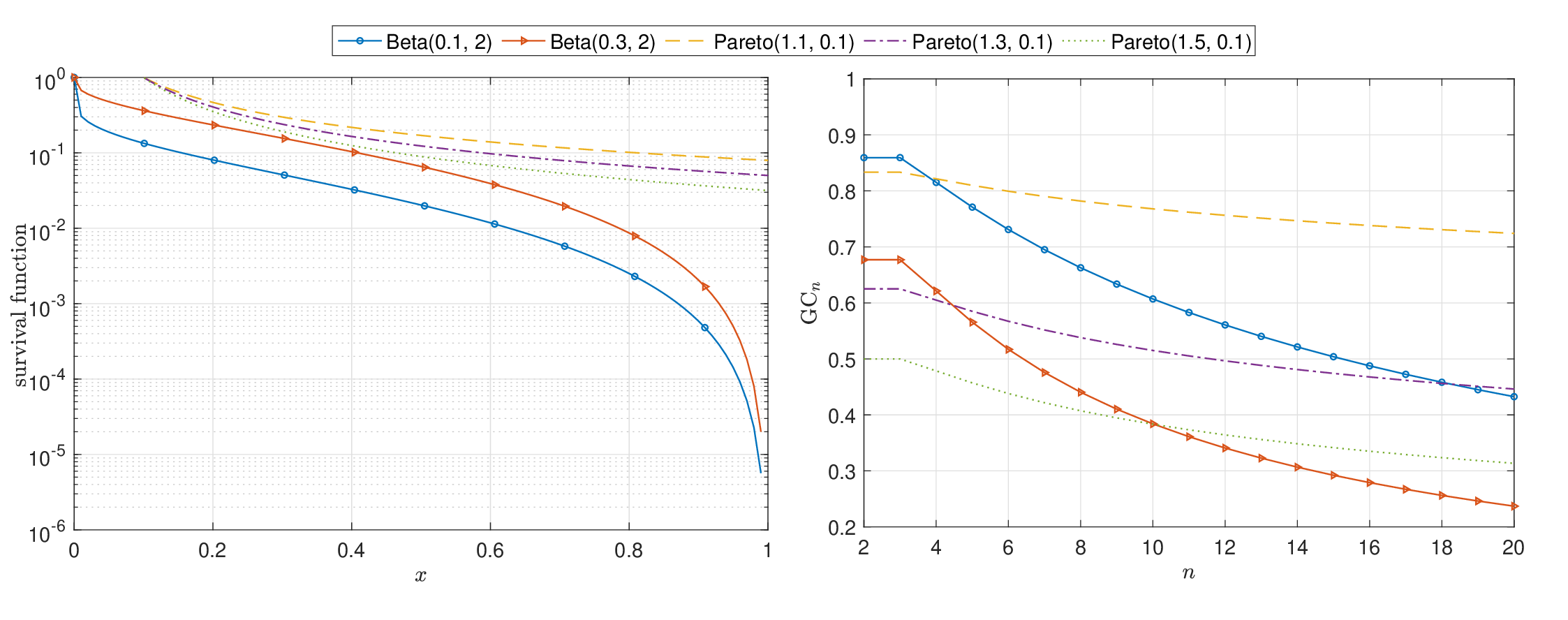}
 \captionsetup{font=small}
 \caption{ \small Survival functions (left panel)  and  values of  $\mathrm {GC}_n$ (right panel) for 
$\mathrm{Beta}(a,2)$ and $\mathrm{Pareto}(\alpha,0.1)$}\label{fig:2}
\end{figure}

On  the left panel of Figure \ref{fig:1}, we show  the survival functions of the Pareto and log-normal distributions. Clearly, the Pareto distribution exhibits a slower tail decay, meaning it has a heavier tail compared to the log-normal distribution. On the right panel  of Figure \ref{fig:1}, we illustrate $\GC_n$ as a function of $n$ for both distributions.  Our findings reveal an interesting pattern:  increasing $n$ enhances the sensitivity of $\GC_n$ to tail risk, capturing the disparity between distributions with different tail behaviors. Specifically, when $n = 2$, $\GC_n$ for the Pareto distribution is lower than that of the log-normal distribution. However, as $n$ increases, $\GC_n$ for the Pareto distribution eventually surpasses that of the log-normal distribution.  This behavior can be understood from the definition of $\GC_n$. The term $\GD_n(X)$  is influenced by higher-order differences in the distribution function, making it more sensitive to tail behavior as $n$ increases. Since the Pareto distribution has a heavier tail than the log-normal distribution, the contribution of extreme values becomes more pronounced for larger $n$, leading to a higher $\GC_n$.  The same phenomenon also occurs for the Pareto and Beta distributions in Figure \ref{fig:2}. 

From these numerical results, we can see that for a moderate choice of $n$ such as $n\le 20$, the value of $\GC_n(X)$ is typically not close to $0$, so that the limiting behavior $\GC_n(X)\to0$ in Proposition \ref{pro:1} does not hurt the practical applicability of $\GC_n$.  
Also note that by combining Theorem \ref{th:chara} and  Proposition \ref{pro:1}, we have   two asymptotic results: 
$$
\lim_{n\to\infty} \sup_{X\in L^1_+} \GC_n (X)= 1 ~~ \mbox{ and } ~~
\sup_{X\in L^1_+}  \lim_{n\to\infty} \GC_n (X)= 0.
$$
That is, the limit and the supremum cannot be interchanged. 
Even for  large values of $n$, there exists some random variables in $L^1_+$ for which $\GC_n$ is not close to $0$, and the index $\GC_n$ can be used to analyze the inequality in some extreme distributions.

\end{document}